\documentclass{vldb}

\usepackage{times, algorithm, algorithmic}
\usepackage{url}
\usepackage{subfigure}
\usepackage{color}
\usepackage{amsmath}
\usepackage{xspace}
\usepackage{ntheorem}
\usepackage{cite}
\usepackage{listings}
\usepackage{upquote} 
\usepackage{color}
\algsetup{linenosize=\small}

\newenvironment{myitem}{
\begin{itemize}
\setlength{\parskip}{0pt}
\setlength{\itemsep}{0pt}
\setlength{\partopsep}{0pt}
\setlength{\parskip}{0pt}
\setlength{\topsep}{0pt}
\setlength{\parsep}{0pt}}{\end{itemize}
}

\newcommand{\ourSketch}{SBG-Sketch\xspace}

\newtheorem*{nonumtheorem}{Theorem}
\newtheorem{theorem}{Theorem}

\begin{document}
\title{
%SBG-Sketch for Archive
\ourSketch{}: A Self-Balanced Sketch for \\ Labeled-Graph Stream Summarization
}

\numberofauthors{1}
\author{
\alignauthor
Mohamed S. Hassan, Bruno Ribeiro, Walid G. Aref\\
\affaddr{Purdue University, West Lafayette, IN, USA}\\
\email{{\large\{}msaberab,~ribeiro,~aref{\large\}}@cs.purdue.edu}
}

\maketitle

%------------------------------------------------
%------------------------------------------------
\sloppypar
\abstract
Applications in various domains rely on processing graph streams, e.g., communication logs of a cloud-troubleshooting system, road-network traffic updates, and interactions on a social network. A labeled-graph stream refers to a sequence of streamed edges that form a labeled graph. Label-aware applications need to filter the graph stream before performing a graph operation. 
Due to the large volume and high velocity of these streams, it is often more practical to incrementally build a lossy-compressed version 
of the graph, and use this lossy version to approximately evaluate graph queries. Challenges arise when the queries are unknown in advance but are associated with filtering predicates based on edge labels.
Surprisingly common, and especially challenging, are labeled-graph streams that have highly skewed label distributions that might also vary over time. 
This paper introduces {\em S}elf-{\em B}alanced {\em G}raph {\em Sketch} (\ourSketch{}, for short), a graphical sketch for summarizing and querying labeled-graph streams that can cope with all these challenges. 
\ourSketch{} maintains synopsis for both the edge attributes (e.g., edge weight) as well as the topology of the streamed graph.
\ourSketch{} allows efficient processing of graph-traversal queries, e.g., reachability queries.
Experimental results over a variety of real graph streams show
\ourSketch{}
to reduce the estimation errors of state-of-the-art methods by up to $99\%$.
%------------------------------------------------

\section{Introduction}
\label{sec:Introduction}
The ubiquity of high-velocity data streams that form graph structures gives rise to many interesting graph queries, and consequently, query processing challenges. A graph stream refers to a data stream of tuples representing graph edges that can form a graph structure. For example, in a cloud-troubleshooting application, a graph stream can be a sequence of edges representing communication logs among the cloud's machines. Each communication log is a directed edge from a sender to a receiver, where edges have labels, e.g., the communication-protocol used. In this paper, we focus on these labeled-graph streams 
%which 
that 
raise interesting data-management challenges.

More precisely, a labeled-graph stream is a graph stream where each edge is associated with a a categorical attribute (label). Associating labels to edges helps in defining and evaluating constrained graph queries, where a query filters the stream using the edge labels before query evaluation. 
Consider the following real-world queries on labeled-graph streams:

\noindent
{\bf $\bullet$~Communication Networks}:
Cloud-environment operators usually analyze the communication-log stream to perform real-time troubleshooting. A typical communication-log entry describes a communication between two machines, namely the source and the destination machines, as well as some communication attributes, e.g., the round-trip time, the sender's application id. This stream forms a labeled-graph stream, where an edge's label is the sender application id. A cloud-troubleshooting application, say $A$, may issue a constrained reachability-query to detect if messages created by Application~$A$ from a source machine reach a destination machine. This reachability query is constrained to use only edges that represent messages created by Application~$A$ (i.e., Label~$A$). 
%\end{myitem}

%\begin{myitem}
%\item {\bf Road networks:} 
%\noindent
%{\bf $\bullet$~Road Networks}:
%A route-planning application may receive a labeled-graph stream that describes traffic updates, e.g., where a streamed edge can be an updated travel-time of a specific road defined by its endpoints as well as the road type (e.g., highway, toll). The route-planning application may analyze the labeled-graph stream of the road network to update its users with better paths in real-time.

\noindent
{\bf $\bullet$~Social Networks}:
A social network may need to detect trending activities of a given object type (e.g., picture, video, status). A graph stream may describe user activities w.r.t.\ these social-network objects, e.g., User $U_i$ shares Post $P_j$. The edge labels represent activity types, e.g., comment, share, or like. The social network may detect trending posts or objects w.r.t.\ specific activity types, e.g., find the most re-shared post.

%\end{myitem}

Labeled-graph streams in applications like the aforementioned ones are usually of large volumes and high velocities.
For example, a cloud-troubleshooting application of a commercial cloud-service receives a labeled-graph stream at the rate of $9$~million edges per second (i.e., a stream query acting on a one-minute window has $0.54$~billion edges to process).
Moreover, low-latency for processing queries becomes necessary in many applications, e.g., when detecting security-threats in real-time. Hence, it is practical to summarize a graph stream by incrementally building a smaller stream-synopsis.
This addresses the data-volume challenge, where bounded-memory is allocated to summarize the continuously-arriving edges of a high-volume graph stream.
In addition, the low-latency requirement may be addressed by approximating the query results instead of producing exact answers.
%In addition, in contrast to producing exact answers to graph queries, in order to have low-latency one may need to produce approximate answers. 

Summarizing {\em labeled}-graph streams has additional challenges. The labels of the streamed edges are unevenly distributed.
Thus, it is common in 
%such
these 
streams to find frequent labels or infrequent ones. The uneven distribution of edge labels may not be known beforehand and may change over time.
For example, a cloud-troubleshooting application will have a communication-log graph stream with edges representing communication types that are more frequent than other types (e.g., HTTP communication-log entries may dominate).
This imbalance raises a challenge in summarizing such a stream, where no edge type (label) should be penalized w.r.t. summarization accuracy due to the rareness or the relative high frequency of its label.

In this paper, we present {\ourSketch}, a graphical sketching technique that automatically balances the sketch load according to the relative frequency of edge labels without penalizing edges with rare labels. Given a labeled-graph stream, say~$G$, and a fixed memory-size, say~$\text{Mem}_\text{max}$, \ourSketch{} uses $\text{Mem}_\text{max}$ memory to incrementally summarize both the attributes of the edges of $G$ as well as the topology of the graph formed by Stream~$G$.

The main idea of \ourSketch{} is to allow edges of 
%high-frequent 
high-frequency 
labels to automatically leverage unused memory previously assigned to 
%low-frequent 
low-frequency
labels with a guarantee that edges of low-frequency labels can use that memory whenever needed in the future.
Notice that bounding the memory allocated to handle a graph stream is important.
The benefit of this memory-bounding is twofold. First, edges arrive continuously with large volume and high velocity in many applications, where storing all 
%of them 
edges
is impractical in many scenarios. Second, query time-efficiency will enhance as queries will process 
a
bounded-synopsis 
%which 
that
is much smaller than the raw graph stream.

The contributions of this paper are as follows:
\begin{myitem}
\item We introduce the design of \ourSketch{}. \ourSketch efficiently summarizes labeled-graph streams and automatically balances sketch load in streams with unpredictable and highly imbalanced edge-label frequencies, all without penalizing edges with rare labels. 

\item The design of \ourSketch{} is of interest on its own as it represents a departure from the count-min design of Cormode and Muthukrishnan~\cite{CountMin_Algo_2005} used in some of the most recent works in graphical sketching~\cite{gSketch_VLDB_2011,TCM_SIGMOD_2016}, and thus enabling new applications. The main focus of this paper is 
%one such new application: 
to enable applications that call for labeled-graph stream sketching for skewed label-distributions.

\item We show that \ourSketch{} can give an approximate answer to graph queries constrained to sets of edge labels. We demonstrate the use of \ourSketch to compute a variety of approximate queries, reachability queries (with no false-negatives), edge count queries, and sub-graph queries. These queries can serve a wide spectrum of applications.

\item We conduct extensive experiments using three real-world datasets from different domains. Results demonstrate that \ourSketch{} can effectively 
%handle the challenges of summarizing 
summarize 
labeled-graph streams, and 
effectively estimate constrained graph-queries. Moreover, we show that \ourSketch{} significantly outperforms the state-of-the-art graph sketch method~\cite{TCM_SIGMOD_2016} w.r.t. estimation accuracy.
\end{myitem}

The rest of this paper proceeds as follows.
Section~\ref{sec:GraphStreamModel} defines the model we follow for graph streams of labeled edges.
Section~\ref{sec:ProblemDefinition} identifies the requirements that an effective sketching method should satisfy 
in order 
to handle labeled-graph streams, while Section~\ref{sec:SolutionApproach} introduces our 
%solution 
approach.
Section~\ref{sec:sketchStructure} presents the structure of \ourSketch{} as well as the general logic for updating the sketch upon edge arrivals.
Section~\ref{sec:queryEstimation} demonstrates how \ourSketch{} can %be used to 
estimate some important constrained-queries on graph streams.
Section~\ref{sec:ExperimentalEvaluation} presents the experimental evaluation of \ourSketch{}.
The related work is discussed in Section~\ref{sec:RelatedWork}. Finally, Section~\ref{sec:Conclusion} contains concluding remarks.

\section{Graph Modeling and Problem Definition}
%\vfill\eject
\subsection{The Graph-Stream Model}
\label{sec:GraphStreamModel}

We model a labeled-graph stream, say $G$, as a data stream of labeled edges $(e_1, e_2,~\ldots,~e_m)$. This graph stream forms Graph~$G~=~(V,~E,~L)$, where $V$ is the vertex set of $G$, $E$ is the edge set formed by the streamed edges, and $L$ is the set of distinct edge labels.
A streamed edge, say $e_i$, is defined as $e_i~=~(s_i,~d_i,~l_i,~w_i)$, where $s_{i}~\in~V$, $d_{i}~\in~V$, $l_{i}~\in~L$, and $w_{i}~\in~Real$, are the source vertex, the destination vertex, the label, and the weight (real number) of Edge~$e_i$, respectively. For simplicity, we assume that the graph edges are directed. However, all the techniques presented in this paper can be applied to undirected graphs.
Figure~\ref{Fig:ExampleLabeledGraph} gives a sample graph-stream of nine labeled-edges being streamed. For example, the edge from Vertex~$g$ to Vertex~$f$ is the result of receiving the following stream element $(g,~f,~R,~1)$.

\begin{figure}[t]
\centering
\includegraphics[width=2.0in]
{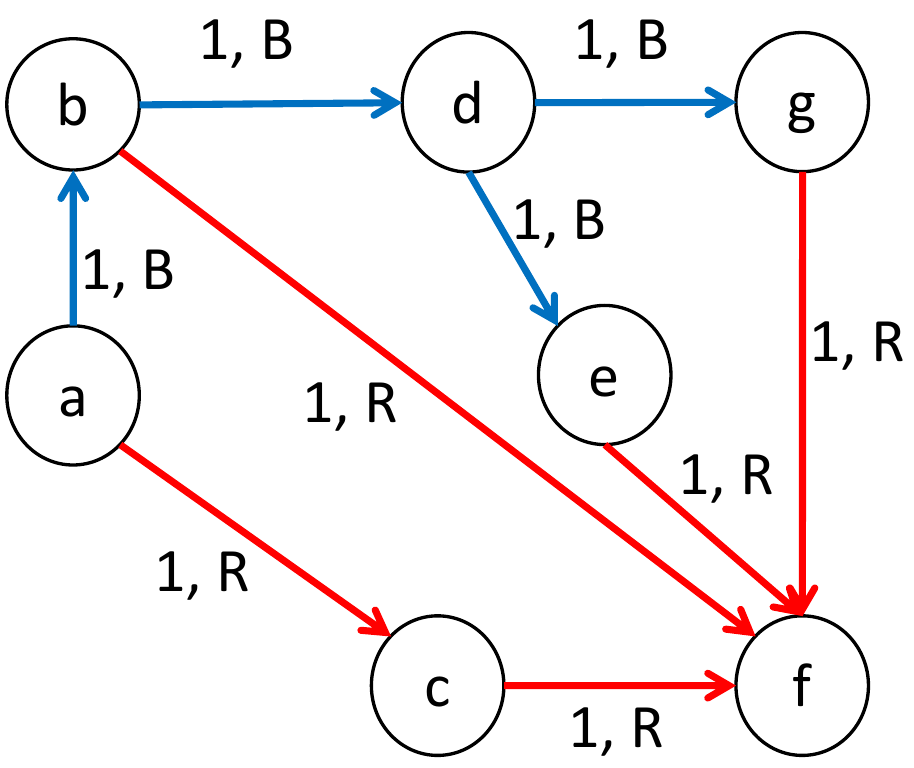}
\caption{Edge-labeled Graph $G$ with each edge having two values: $Weight,~Label$. $R$, and $B$ are two different labels.}
\label{Fig:ExampleLabeledGraph}
\end{figure}

\subsection{Problem Definition and Solution Requirements}
\label{sec:ProblemDefinition}

Given a labeled-graph stream, say $G$, where 
$G = \langle(a_1,~b_1,~l_1,~w_1),~\ldots,~(a_m,~b_m,l_m,w_m)\rangle$, the number of distinct edge-labels, say $L$, and a memory-size upper-bound, say $\text{Mem}_\text{max}$, 
we need to 
construct a graphical sketch, say $S(G)$, that satisfies the following requirements:

%\begin{myenum}
%\item
\noindent
{\bf $\bullet$}~Construct an in-memory synopsis that does not exceed $\text{Mem}_\text{max}$ of memory.

%\item 
\noindent
{\bf $\bullet$}~Summarize the edge weights of Stream~$G$ using an aggregate function defined by the application.

%\item
\noindent
{\bf $\bullet$}~Summarize the topology of Stream~$G$ to support graph-traversal queries (e.g., reachability estimation).

%\item 
\noindent
{\bf $\bullet$}~Consider the edge labels in the summary to support constrained graph-queries.

%\item 
\noindent
{\bf $\bullet$}~Consider the imbalance in the distribution of edges w.r.t labels.
An edge-label should be allowed to use larger quota from the allocated fixed-memory if its edges are more frequent than 
those corresponding to 
other labels. 
This requirement aims that all the edges win or achieve high summarization-accuracy regardless of their label rareness or popularity.

%\end{myenum}

The last requirement is important in real-world scenarios, where the edges are unevenly distributed w.r.t.\ their labels.
For instance, consider a cloud-troubleshooting application, where the streamed edges are labeled by application identifiers. The messaging frequency of some applications can be much higher than those of other applications.
Hence, a sketching method handling this graph-stream model should not penalize the accuracy of summarizing edges due to the rareness of some labels. Moreover, avoiding this penalization should consider using memory wisely (e.g., avoid allocating exclusive large-memory shares to less-representative labels).

\section{Overview of \ourSketch{}}
\label{sec:SolutionApproach}

\begin{figure}[t]
\centering
\includegraphics[width=2.4in]{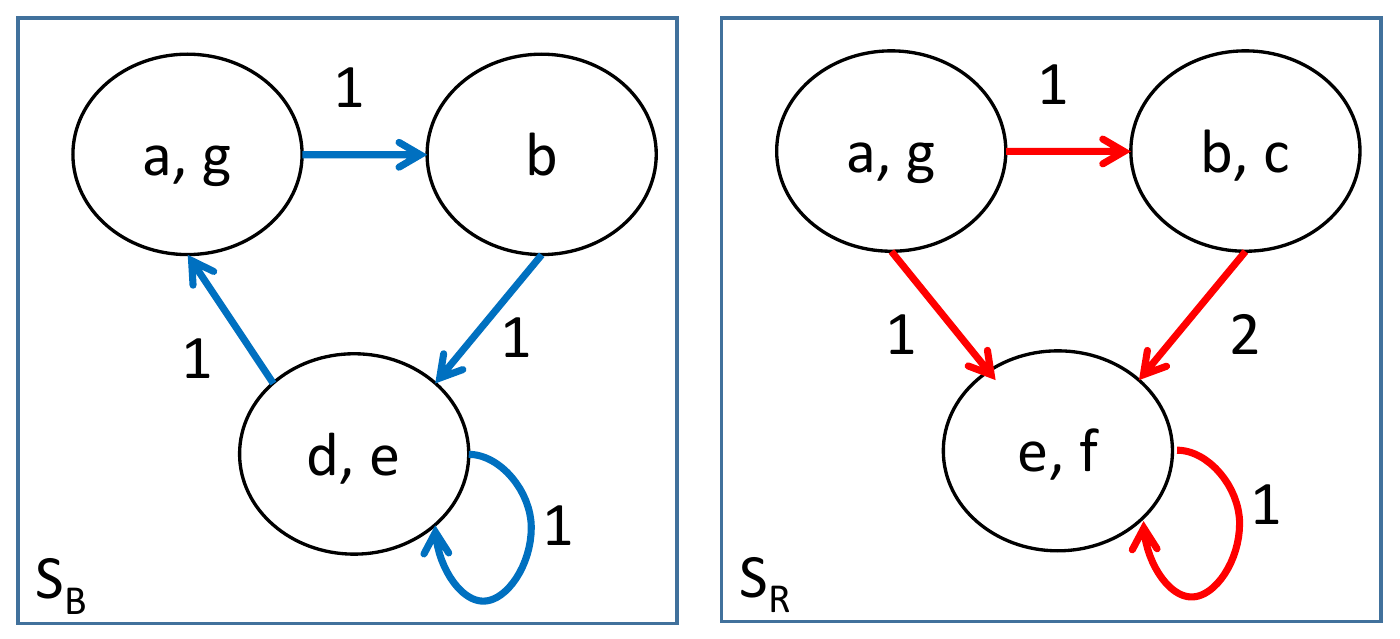}
\caption{A label-aware summary of the Graph in Figure~\ref{Fig:ExampleLabeledGraph}.}
\label{Fig:LabelAwareSketch}
\end{figure}

Given a graph stream as defined in Section~\ref{sec:GraphStreamModel}, our %solution 
approach is to build a sketch that satisfies the requirements stated in Section~\ref{sec:ProblemDefinition}. To illustrate, consider the sample graph stream $G$ in Figure~\ref{Fig:ExampleLabeledGraph}. The proposed \ourSketch graphical sketch follows the structure that Figure~\ref{Fig:LabelAwareSketch} illustrates, where we assume that the sketch is built to aggregate the weights of the received edges by summing them (other aggregates are possible). For each distinct edge-label, say $l$, we allocate a sub-sketch 
$S_l$ 
%to  summarize
that summarizes
the sub-graph of all the edges of Label~$l$.
%, namely, $S_l$. 
For instance, Figure~\ref{Fig:LabelAwareSketch} gives two sub-sketches, namely $S_B$ and $S_R$, that summarize the $Blue$, and the $Red$ edges, respectively. This allows the graphical sketch to evaluate label-constrained queries. For instance, a query allowing only $Blue$ edges will only consult Sub-sketch~$S_B$. Notice that the total size of the sketch is upper-bounded by the maximum memory-size defined by the user that affects the size of each sub-sketch.

The idea of the graphical sketch in Figure~\ref{Fig:LabelAwareSketch} is to build a sub-graph for each edge label, say~$l$, by compressing the edges of Label~$l$. In particular, each sub-sketch, say~$S_l$, has a maximum number of vertexes, say $v(S_l)$, that is smaller than the number of vertexes of the original graph stream. The graphical sketch uses a hash function to project the vertexes of the original graph-stream to the vertexes in a sub-sketch. For instance, the hash function groups Vertexes~$a$ and~$g$ 
%together 
in both $S_B$ and $S_R$ (assuming that both sub-sketches use the same hash function). To illustrate how the edge weights are aggregated, consider the arrival of Edge~$E_1=(b,~f,~R,~1)$, and Edge~$E_2=(c,~f,~R,~1)$, where they affect only Sub-sketch $S_R$ as they are both red edges. Assume that the vertex-mapping hash function groups Vertex~$b$ and Vertex~$c$ into one bucket, and Vertex~$e$ and Vertex~$f$ in another bucket (i.e., the same sub-sketch vertex).
When processing Edge~$E_1$, an edge of weight~$1$ will be created in Sub-sketch $S_R$ between Vertexes~$b$ and~$f$. Then, when inserting Edge~$E_2$, the sub-sketch edge that has been created by $E_1$ will have its weight incremented by one (i.e., accumulating the weight of $E_2$). The reason is that the start vertexes of both $E_1$ and $E_2$ are mapped together, and similarly for their end vertexes.

Observe that the graphical sketch in Figure~\ref{Fig:LabelAwareSketch} summarizes both the edge weights as well as the graph-stream topology. For instance, a query asking for the weight of $Edge(a,~b,~B)$ can be evaluated by consulting Sub-sketch~$S_B$ by hashing the endpoint vertexes of the query, mapping them to vertexes in $S_B$, and retrieving the weight of the corresponding edge in $S_B$. Also, the graphical sketch summarizes the topology of the graph stream to allow graph-traversal queries. For instance, a reachability query inquiring if Vertex~$d$ is reachable from Vertex~$a$ using only $Blue$ edges %can be evaluated to
evaluates to 
true because there is a path connecting the two vertexes in Sub-sketch~$S_B$ (i.e., $(a,~g)~\rightarrow~(b)~\rightarrow~(d,~e)$). Another example is a pattern query that estimates if there is a path of two edges from Vertex~$a$ to Vertex~$f$, where the first edge is of Label~$B$, and the second is of Label~$R$. This is possible by expanding Edge~$(a,~g)~\rightarrow~(b)$ by checking the outgoing edges from Vertex~$b$ in the $Red$ sub-sketch and discovering Edge~$(b,~c)~\rightarrow~(e,~f)$. This forms a positive answer 
%as 
because 
a path of two valid connected-edges from the two sketches satisfy the query.

\section{The Design of \ourSketch{} }
\label{sec:sketchStructure}

In this section, we highlight the general structure of \ourSketch{}. Given a labeled graph, say $G$, we create an \ourSketch{} instance, 
say say $S_G$, 
that summarizes Graph $G$. Sketch $S_G$ considers the topology of 
%Graph 
$G$ so that approximating graph-traversal queries becomes possible. 
Assume that %Graph 
$G$ has $n$ vertexes, $m$ edges, and $L$ distinct edge labels. Then, we create Sketch $S_G$ that has $L$ matrices, where each matrix is a $d \times d$ two-dimensional matrix. Notice that $d$ is much smaller than $n$, the number of vertexes in the graph stream.
Also, $L$ is much smaller than $n$ and $m$ in real labeled-graphs, e.g., the number of interaction types in a protein-interaction network is much smaller than the number of proteins.
This forms a three-dimensional matrix of dimensions $L \times d \times d$. 
Notice that 
it is possible 
%also 
to create multiple independent sketches to summarize Graph~$G$ for better accuracy 
%as 
(see Section~\ref{sec:ExpEdgeQueries_VaryHash}). 
%shows.

Figure~\ref{Fig:GeneralSketchStructure} illustrates the general structure of 
%our sketching method 
\ourSketch{}, 
where $P$ independent sketches can be created to summarize a graph.
For illustration,
%purposes, 
assume that $P~=~1$ (i.e., we have only one sketch).

Consider \ourSketch{} $S_G$ for Graph~$G$. Each cell in $S_G$ maintains an aggregate for a set of streamed edges as Figure~\ref{Fig:GeneralSketchStructure} illustrates. The maintenance of this aggregate may differ based on the query type that $S_G$ is supposed to answer (e.g., a counter to answer edge-frequency queries). An incoming edge is hashed into one of the cells in $S_G$ as 
%we explain in 
explained in
Section~\ref{sec:MappingEdgesToCells}. Notice that if multiple sketches are used, each sketch will have a different hash function to hash the vertexes. 
Observe that each cell holds a pair of values, namely rank and aggregate. Section~\ref{sec:RankingLogic} elaborates on how the rank values are used, while Section~\ref{sec:queryEstimation} shows how the aggregate values are maintained for 
%different 
various 
query types. In the next section, we focus on how the streamed edges are mapped to sketching cells. 

\begin{figure}[t]
\centering
\includegraphics[width=3.4in]
{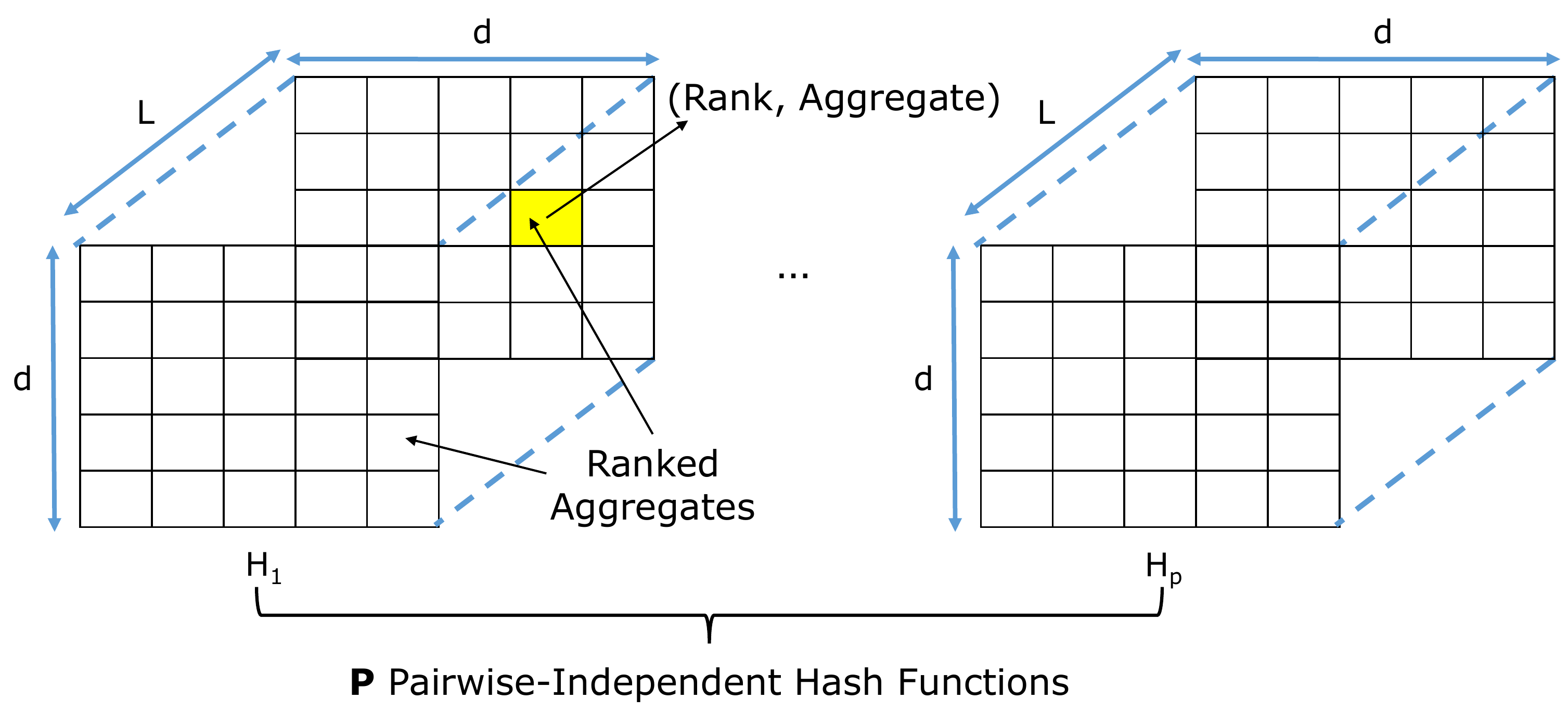}
\caption{The general structure of \ourSketch{}.}
\label{Fig:GeneralSketchStructure}
\end{figure}

\subsection{Mapping Streamed Edges to Sketching Cells}
\label{sec:MappingEdgesToCells}
Mapping a streamed edge to a sketching cell is a fundamental operation to update the sketch. Mapping an edge to a cell is orthogonal to the sketch update logic that depends on the query type supported by the sketch. To illustrate how streamed edges are mapped to cells in \ourSketch, 
%please 
refer to Figure~\ref{Fig:GeneralEdgeToCellMapping} that shows a single sketch $S_G$.
\ourSketch{} generates and uses a set of hash functions. One of these hash functions, namely $H_v$, maps each vertex identifier to an integral value in the range [0,~d-1], i.e., $H_v$ can map any vertex to a row or column of any matrix of the $L$ matrices of Sketch $S_G$. If multiple sketches, say $P$ sketches, are used, then $P$ different pairwise-independent hash functions are generated and are used (i.e., one hash function 
%for each 
per
sketch). To allow traversing the matrices of a single sketch efficiently, \ourSketch uses the same hash function $H_v$ in all the matrices.

Recall that the number of matrices in a sketch is equal to the number of distinct labels of a graph, and we assume that the distinct edge-labels are known beforehand (e.g., the different types of social relationships in a social network).
Refer to Figure~\ref{Fig:GeneralEdgeToCellMapping}. An incoming stream Edge $E$~=~$(a,~b,~l)$ is mapped as follows. First, \ourSketch{} has a static one-to-one-mapping for each label to a corresponding matrix in the Sketch. So, Edge~$E$ will be mapped to one of the cells in the matrix corresponding to Label~$l$, say $M_l$. Hash-function $H_v$ maps Source-vertex~$a$, and Destination-vertex~$b$ to a row, and a column in Matrix~$M_l$, respectively. So, the cell corresponding to Edge~$E$ is Cell ($H_v(a)$, $H_v(b)$) in Matrix~$M_l$, or Cell ($H_v(a)$, $H_v(b)$, $l$), for short. Notice that using the same hash function in all the rows and columns of all the matrices allows traversing the matrices of the sketch efficiently, otherwise, materializing the hash functions as in~\cite{TCM_SIGMOD_2016} would be necessary and additional memory would be consumed from the allocated memory.
For example, to traverse the outgoing edges of the end-vertex of Edge~$E$ (i.e., Vertex~$b$), where the edges are labeled by Label~$i$, we can check the second row of Matrix $M_i$. The reason is that Vertex~$b$ has been mapped by Function $H_v$ to the second column as Figure~\ref{Fig:GeneralEdgeToCellMapping} illustrates, and that all the matrices are adjacency matrices using the same $H_v$ function.

\begin{figure}[t]
\centering
\includegraphics[width=2.6in]
{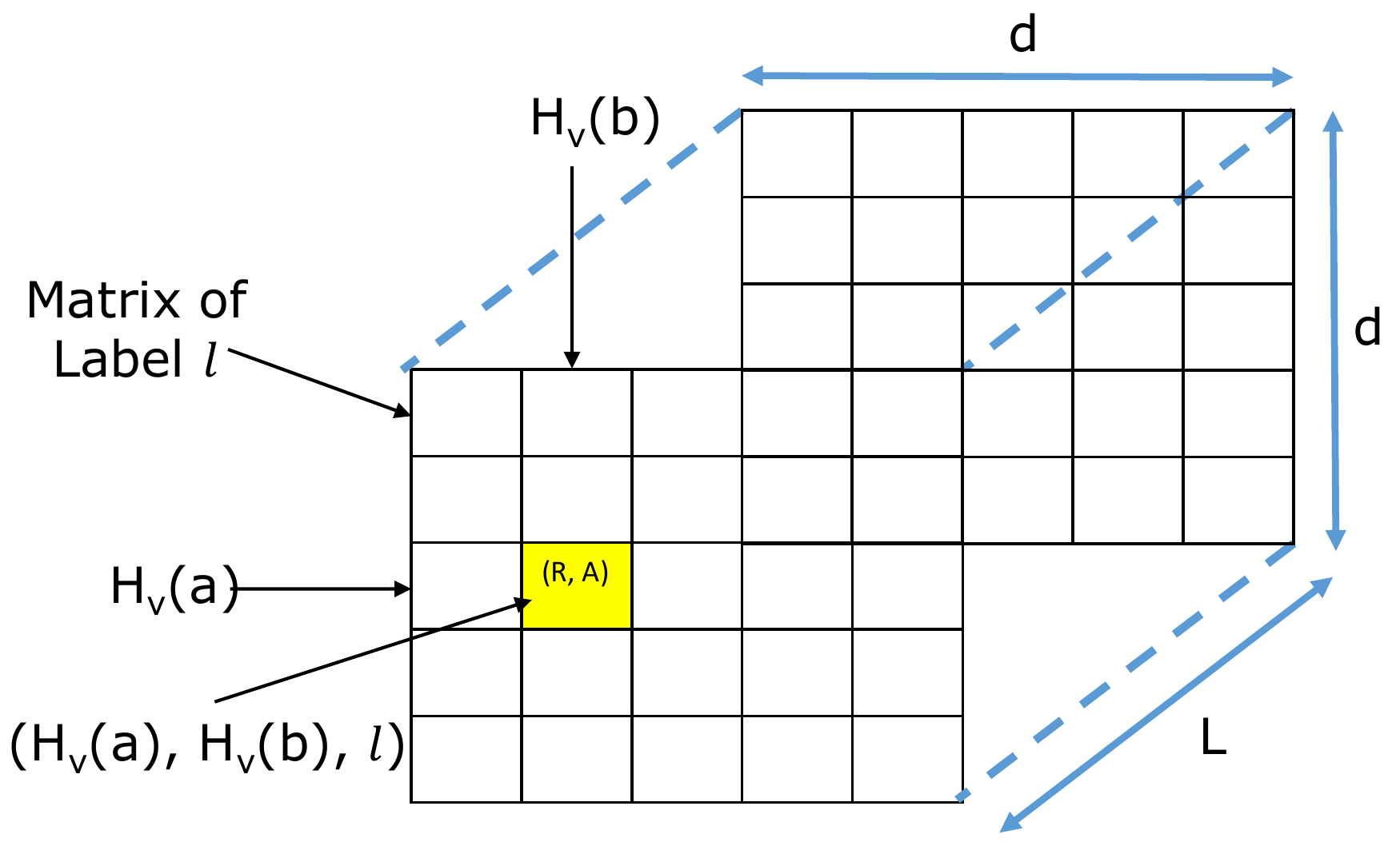}
\caption{Mapping a labeled-edge stream input to a cell in \ourSketch{}.}
\label{Fig:GeneralEdgeToCellMapping}
\end{figure}

\subsection{The Ranking Logic in \ourSketch}
\label{sec:RankingLogic}

Usually, 
edge-labeled graphs are 
%usually 
skewed in numbers and are unbalanced w.r.t. the frequency of edges per distinct edge-label. For instance, in a social network, the number of family-type relationships may be much less than the friend-type relationships.
This adds a challenge when building graph sketches for edge-labeled graphs. In particular, memory for summarizing edges of a specific label, say $l$, should be proportional to the frequency of receiving edges of Label~$L$. Otherwise, precious parts of the sketch would be wasted (i.e., those matrices corresponding to low-frequency labels would have wasted cells). The challenge becomes more difficult when the frequency of labels is not known beforehand. 
In this case, 
%, where even 
initializing the matrices with different dimensions may become difficult or inaccurate.

\ourSketch{} addresses this challenge without requiring to know the relative edge-label frequencies beforehand. The main idea of \ourSketch{} is to allocate matrices of the same dimensions to all the labels, and to allow an edge, say $E~=~(a,~b,~l_0)$, of Label~$l_0$ to use cells in matrices that do not correspond to Label~$l_0$. The intuition 
%of this 
behind this approach
is to allow high-frequent labels (e.g., $l_0$) to use other matrices corresponding to 
%low-frequent 
low-frequency
labels. However, \ourSketch{} guarantees that the low-frequency labels can reclaim their cells that were occupied by high-frequency labels whenever needed. To illustrate, consider an edge-labeled graph, say $G_R$, with a total of five labels, i.e., L~=~5. Let $S_R$ be \ourSketch{} for Graph $G_R$, where $S_R$ consists of five matrices, namely, $M_{0}$, $M_{1}$, $M_{2}$, $M_{3}$, and $M_{4}$ as in Figure~\ref{Fig:GeneralRankingLogic}.
Upon receiving Edge~$E$, \ourSketch{} assigns a rank vector to Edge~$E$ before updating the sketch. In Section~\ref{sec:RankGeneration}, we discuss how the rank vectors are generated and 
are 
assigned to edges. For now, it is sufficient to know that the values of a rank vector 
%is 
are
a permutation of the values $\{0, 1, ..., |L|-1\}$, where $|L|$ is the number of matrices in the sketch (i.e., the number of labels). For example, Figure~\ref{Fig:GeneralRankingLogic} illustrates that Edge~$E$ is assigned a ranking Vector, say $RV_E$, of values [0, 2, 1, 4, 3], where 0 is the highest rank, and 4 is the lowest rank.

An element, say $RV[i]$, of a rank-vector for an edge determines the rank of the edge in Matrix~$M_i$. For instance, Figure~\ref{Fig:GeneralRankingLogic} illustrates that the rank of Edge~$E$ in Matrix $M_1$, i.e., $RV_{E}[1]$, is equal to 2.
Notice that each cell, say $C$, in the sketch stores a rank value, say $R(C)$. Rank value $R(C)$ represents the rank of the last edge that has updated Cell~$C$. For example, in Figure~\ref{Fig:GeneralRankingLogic}, the yellow cell in the top-left Matrix $M_0$ has a rank value of~1, which means that the last edge that has updated this cell has a rank value of 1 in Matrix~$M_0$. Whenever an edge, say $E$, is hashed into a cell, say $C$, the rank of Cell~$C$ as well as the rank of Edge~$E$ in the matrix hosting Cell~$C$ 
%determine 
determines
if Edge~$E$ can use Cell~$C$.
%or not. 
In particular, a streamed edge can use and affect Cell~$C$ if and only if the edge's rank is higher than or equal to the current rank of Cell~$C$. Comparing rank value $RV_{E}[i]$ (i.e., the rank for Edge~$E$ in Matrix~$M_i$) to the rank value of a cell in Matrix~$M_i$, say $R(C)$, leads to the following three cases with three possible outcomes (notice that zero is the highest rank --  refer to Figure~\ref{Fig:GeneralRankingLogic} for illustration):

%\begin{itemize}
\noindent
{\bf $\bullet$~Evict and Occupy:}
%\item {\bf Evict and Occupy:}~
If $RV_{E}[i]$~\textgreater~$R(C)$, then evict the effect of all the edges that have affected Cell~$C$, use $C$ to update Matrix~$M_i$ by the arrival of Edge~$E$, and set the rank of $C$ to the value of $RV_{E}[i]$. This prevents any edge of rank lower than $RV_{E}[i]$ to evict Edge~$E$ from Cell~$C$. For instance, in Figure~\ref{Fig:GeneralRankingLogic}, Edge~$E$ has a higher rank in Matrix~$M_0$ than the last edge that has contributed to Cell~$C$ in Matrix~$M_0$. Thus, the aggregate value in Cell~$C$ is replaced by the value associated with Edge~$E$ (e.g., may be set to 1 if the sketch is counting the frequency of receiving Edge~$E$), and the rank of Cell~$C$ is set to $RV_{E}[0]$, which is zero in this example.

\noindent
{\bf $\bullet$~Update the Aggregate:}
%\item {\bf Update the Aggregate:}~
If $RV_{E}[i]$~=~$R(C)$, then update the aggregate of Cell~$C$, and leave the rank of Cell~$C$ unchanged. This preserves the aggregation of the previous instances of this edge or other edges of the same rank that collide with Edge~$E$ in Cell~$C$. For example, in Figure~\ref{Fig:GeneralRankingLogic}, Edge~$E$ has the same rank in Matrix~$M_2$ as the rank of the last edge that has contributed to Cell~$C$. Thus, the aggregate value in Cell~$C$ is updated by the value associated with Edge~$E$, e.g., may be incremented by one if the sketch is counting the frequency of receiving Edge~$E$.

\noindent
{\bf $\bullet$~Do Nothing:}
%\item {\bf Do Nothing:}~
If $RV_{E}[i]$~\textless~$R(C)$, then do nothing to Cell~$C$. This means that the last edge, say $E_{last}$, that has contributed to Cell~$C$ has %had 
a higher rank that prevents Edge~$E$ from evicting Edge~$E_{last}$ or even contributing to Edge~$E_{last}$'s aggregate value. For example, in Figure~\ref{Fig:GeneralRankingLogic}, Edge~$E$ has a lower rank in Matrix~$M_4$ than 
that of the last edge that has contributed to Cell~$C$. Thus, 
the value in
Cell~$C$ is kept unaffected.
%\end{itemize}  

\ourSketch{} guarantees that all edges of Label~$X$ have the highest priority in the matrix corresponding to Label~$X$, i.e., Matrix~$M_X$. This guarantees that all the edges with Label~$X$ have the highest rank of zero in their ``home" matrix $M_X$. Hence, an edge of Label~$Y$, where $Y~\neq~X$, can possibly use a cell, say $C_\text{rented}$, in Matrix $M_X$, if $C_\text{rented}$ has never been used by an edge of Label~$X$. Moreover, edges with Label~$X$ are given the privilege to evict lower-ranked edges in Cell~$C_\text{rented}$, use the cell, and disallow any edge not labeled by Label~$X$ to use Cell~$C_\text{rented}$. This is achieved by updating the rank of Cell~$C_\text{rented}$ with zero, the highest rank that cannot be evicted.    

Notice that any query, say~$Q$, processing Edge~$E$, should consult only the cells that hold the ranks of Edge~$E$. For example, Query~$Q$, regardless of its type, when retrieving Edge~$E$ from \ourSketch{}, will consult only matrices $M_0$, $M_1$, and $M_2$ in Figure~\ref{Fig:GeneralRankingLogic}. The reason is that the values at these cells may represent contributions by Edge~$E$. However, Matrices~$M_3$ and~$M_4$ should not be considered when querying Edge~$E$ as their ranks guarantee that Edge~$E$ has not contributed to their current aggregate values, otherwise, they would hold the ranks corresponding to the ranks of Edge~$E$.
Notice that \ourSketch{} does not allow an edge to use more than one cell per matrix. Using more than one cell per matrix 
%will 
would
increase the processing time as well as the collision rate, which may decrease the approximation accuracy. However, using one cell per matrix gives each edge a chance to use a cell that might be unoccupied in each matrix.

\begin{figure}[t]
\centering
\includegraphics[width=3.3in]
{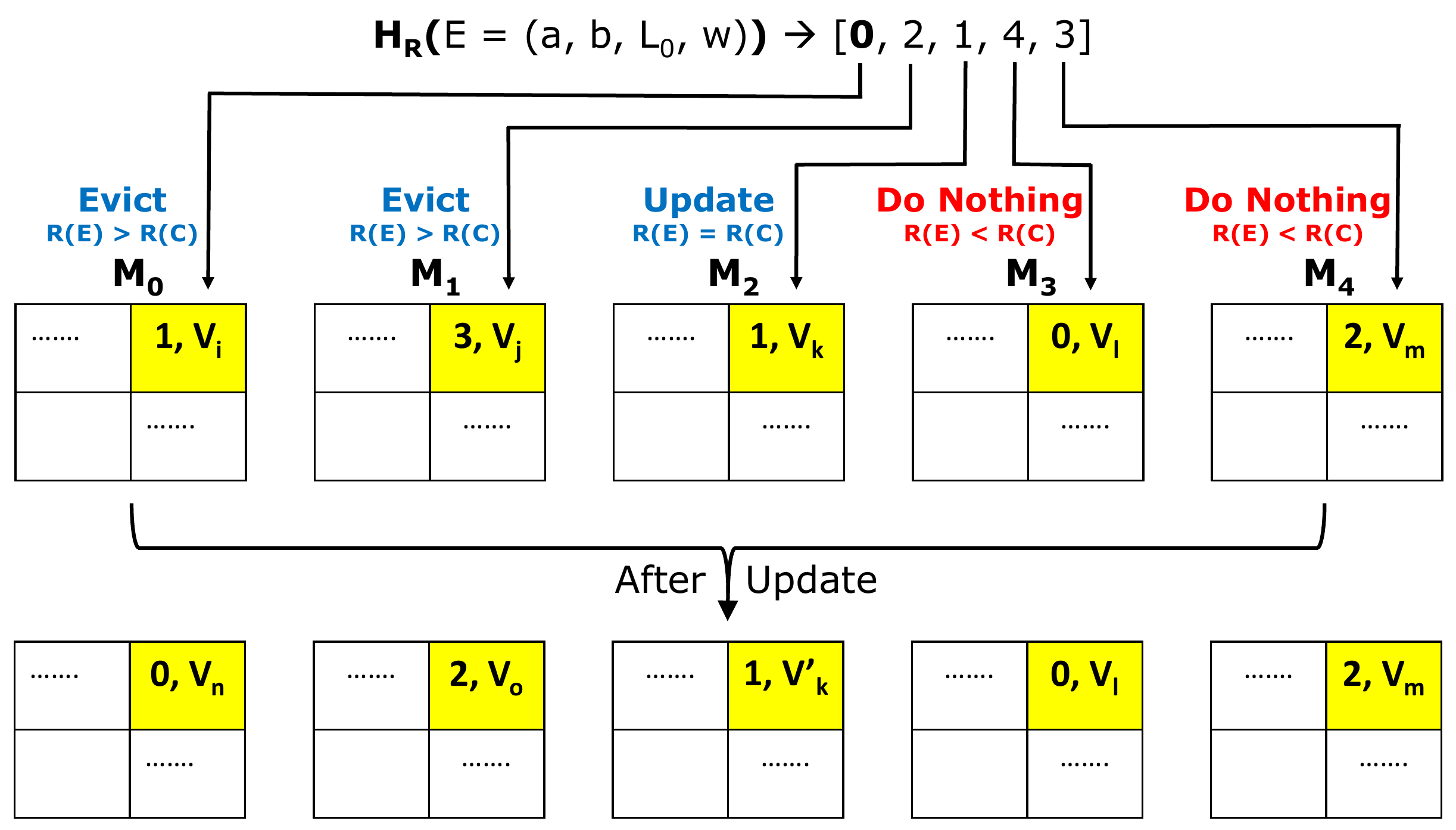}
\caption{An edge is allowed to use cells in matrices of other labels %matrices 
if the edge rank is higher than or equal to the cells' ranks.}
\label{Fig:GeneralRankingLogic}
\end{figure}

\subsection{Generation and Mapping of Rank Vectors}
\label{sec:RankGeneration}
In this section, we describe how \ourSketch{} generates the rank vectors as well as how a rank vector is assigned to an edge. Recall that a rank vector, for a given Graph $G$, is a permutation of the integer values $\{0, 1, ..., |L|-1\}$, where $|L|$ is the number of distinct edge-labels of Graph~$G$. Also, recall that zero is the highest rank. \ourSketch{} accepts an initialization parameter, namely $P_{ranks}$, that corresponds to the number of distinct random rank-vectors that \ourSketch{} generates and uses. \ourSketch{} restricts that the number of rank vectors is upper-bounded by the factorial of $|L|-1$, i.e., $P_{ranks}~\leq~(L-1)!$. This restriction makes it possible to generate $P_{ranks}$ rank vectors that are all unique.

For illustration, refer to an example for generating rank-vectors in Figure~\ref{Fig:RankVectorGeneration}. 
In the figure, we assume that \ourSketch{} is initialized for Graph~$G_R$ that has five distinct edge-labels (i.e., $|L|$~=~5), and that the number of the rank vectors to generate is four (i.e., $P_{rank}$~=~4). Using these parameters, \ourSketch{} generates and materializes four random rank-vectors that are different permutations of the values $\{1, 2, ..., |L|-1\}$. Notice that zero is not considered in the materialized rank-vectors as in Figure~\ref{Fig:RankVectorGeneration}. Rank zero is injected on the fly into a rank vector when that vector is selected for an incoming edge, where the injection position is controlled by the edge's label.

To illustrate how a rank vector including zero is assigned to a streamed edge, assume that \ourSketch{} receives Edge~$E~=~(a,~b,~L_0)$ as 
%illustrated 
in Figure~\ref{Fig:RankVectorGeneration}. \ourSketch{} uses a hash function, namely $H_R$, that hashes Edge~$E$ into a value in the integral range $[0,~P_{ranks}~-~1]$. In the example in Figure~\ref{Fig:RankVectorGeneration}, Function~$H_R$ accepts the source vertex, the destination vertex, and the label of Edge~$E$ as inputs to hash Edge~$E$ into either 0, 1, 2, or 3. In this example, Edge~$E$ is assigned $RV[1]$ as its rank vector. However, $RV[1]$ does not include Rank zero that defines the matrix where Edge~$E$ has the highest rank. \ourSketch{} uses the label of Edge~$E$ to augment the selected rank-vector with the zero rank-value. This augmentation assures that Edge~$E$ has the highest rank in the matrix corresponding to the label of Edge~$E$. For example, as the label of Edge~$E$ is $L_0$, \ourSketch{} injects zero into the first element in the generated rank-vector, i.e., the assigned rank-vector becomes $[0, 2, 1, 4, 3]$. Notice that if Edge~$E$ had another label, say $L_1$, then Rank zero would be injected in the location corresponding to Matrix~$M_1$.

\begin{figure}[t]
\centering
\includegraphics[width=2.7in]
{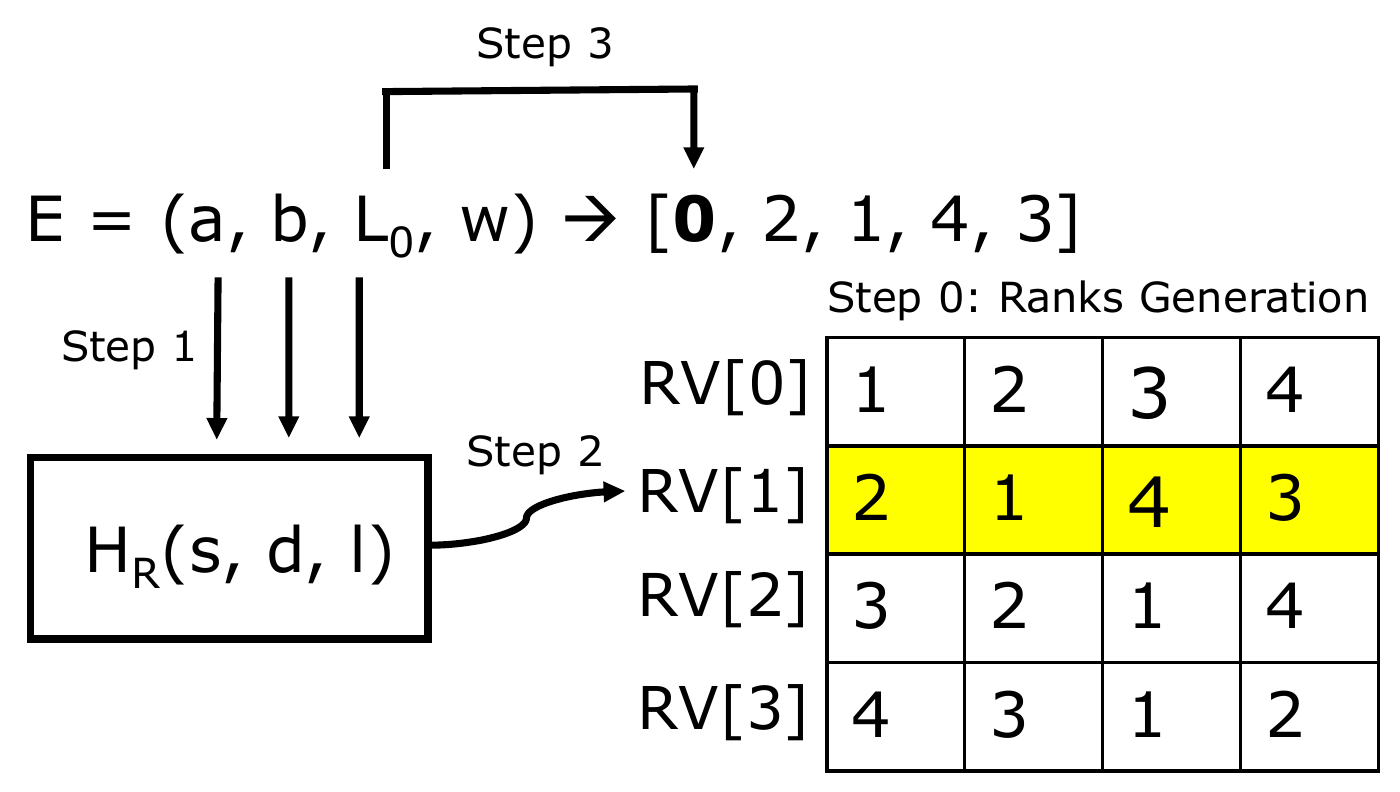}
\caption{An edge is assigned the highest rank, i.e., zero, according to the edge's label in the corresponding sketch matrix.}
\label{Fig:RankVectorGeneration}
\end{figure}

\section{Query Estimation}
\label{sec:queryEstimation}

In this section, we describe how \ourSketch{} estimates the results of various query types. In particular, Section~\ref{sec:freqQueries} elaborates on frequency-based queries (e.g., the constrained edge-frequency query), and Section~\ref{sec:traversalQueries} describes how \ourSketch{} estimates graph-traversal queries (e.g., the constrained reachability query). For each query type, we demonstrate how the sketch is updated when receiving a streamed edge as well as how the sketch is queried to evaluate the query approximately.

\subsection{Frequency-Based Queries}
\label{sec:freqQueries}
For frequency-based queries, we assume that a streamed edge is associated with an attribute, say $weight$, of a numerical type that can be aggregated. Without loss of generality, we term this query type a {\em frequency-based} query. Applications usually use this query type to estimate the occurrence frequency of a given edge or sub-graph in a stream. %However, other aggregates are possible as we demonstrate in Section~\ref{sec:spQueries}. 

\subsubsection{Edge Queries}
\label{sec:edgeQueries}
Given two vertex identifiers $a~\in~V$ and $b~\in~V$ and an edge label $l~\in~L$, 
%we denote 
let 
$f_{e}(a,~b,~l)$ 
%as 
be
the exact aggregated edge-weight from Vertex~$a$ to Vertex~$b$, where the edge is labeled by Label~$l$. 
Furthermore, let $\hat{f}_{e}(a,~b,~l)$ be 
%We denote 
the estimated weight of this edge.
%by $\hat{f}_{e}(a,~b,~l)$.

Query~$\hat{f}_{e}(a,~b,~l)$ represents an edge-query. For instance, in a social network, one may estimate the number of comments from User~$A$ on a post by User~$B$, where a comment is represented as a directed edge from User~$A$ to User~$B$ with an edge-label 
%of 
``comment" (other interactions could be represented by other edge labels).

\noindent
{\bf Insertion of Edges:} Algorithm~\ref{Algo:UpdateEdgeQuerySketch} depicts how \ourSketch{} is maintained when inserting an edge to estimate later edge queries. Refer to Figure~\ref{Fig:EdgeCountStreaming} for illustration. Assume that an instance of \ourSketch{} is built for processing a graph of five labels, and is receiving Edge~$E = (a,~b,~L_0,~1)$ with Rank-vector $RV(E)~=~[0,~2,~1,~4,~3]$. Assume further that the sketch is built to perform a sum aggregation on the weight attribute that is set to one for Edge~$E$. Figure~\ref{Fig:EdgeCountBefore} gives \ourSketch{} just before receiving Edge~$E$, where each cell, say $C$, holds an aggregate corresponding to the weights of some aggregated edge weights, and the rank of the last edge that has contributed to Cell~$C$. For instance, in Figure~\ref{Fig:EdgeCountBefore}, the highlighted cell in Matrix~$M_2$ illustrates that the aggregated sum in the cell is 4, and that the last edge with the highest rank that has contributed to this cell has Rank~1 for Matrix~$M_2$. To update \ourSketch{} with Edge~$E$, the corresponding rank vector of Edge~$E$ is computed as illustrated in Section~\ref{sec:RankGeneration}.
Then, the cells that are potential candidates for 
use by Edge~$E$ 
%to use 
are selected. In particular, the cell corresponding to $(H_v(a),~H_v(b))$ in each matrix is a potential candidate (Figure~\ref{Fig:EdgeCountBefore} highlights these cells in yellow).

As an optimization, the cells corresponding to $(H_v(a),~H_v(b))$ in each matrix are physically stored in contiguous memory, thus exhibiting high locality of memory access (i.e., the matrices 
%shown 
given
in Figure~\ref{Fig:EdgeCountBefore} are a logical representation of a single larger physical-matrix).
As explained in Section~\ref{sec:RankingLogic}, according to the rank values in the potential candidate cells and the ranking vector of Edge~$E$, only a subset of these
potential cells may be updated by Edge~$E$ (we term these cells {\em candidate cells}). Figure~\ref{Fig:EdgeCountAfter} illustrates that Edge~$E$ evicts the value at Matrix~$M_0$ because Edge~$E$ has the highest rank in Matrix~$M_0$ (Lines$7-9$ in Algorithm~\ref{Algo:UpdateEdgeQuerySketch}). Eviction also happens in Matrix~$M_1$. However, in Matrix~$M_2$, the ranks are equal, so Edge~$E$ increments the aggregate value of the 
corresponding cell 
%there 
(Lines$10-11$ in Algorithm~\ref{Algo:UpdateEdgeQuerySketch}). For the last two matrices, the cells are occupied by edges with higher ranks, so Edge~$E$ is
prevented
from using these cells. 

Notice that we update the cell of Matrix~$M_2$ in Figure~\ref{Fig:EdgeCountAfter} for illustration purposes only. However, as an accuracy optimization, \ourSketch{} does not need to update that cell and will leave its value to be 4. The reason is that, in this case, \ourSketch{} can guarantee that Edge~$E$ has never been received before. Otherwise, the candidate cell in Matrix~$M_0$ of Figure~\ref{Fig:EdgeCountBefore} would have Rank zero if Edge~$E$ has been encountered before. Hence, the cell in Matrix~$M_2$ does not need to be incremented. 
The reason is that 
%because,
after updating the sketch with $E$'s arrival, the value 4 in Matrix~$M_2$ will implicitly count the arrival of $E$ without any 
%increment in 
changes to
Matrix~$M_2$. 
%So, 
Thus, the value of the candidate cell in Matrix~$M_2$ of Figure~\ref{Fig:EdgeCountAfter} will be kept 
unchanged (i.e., with value 4)
%and not 5 
to help 
increase
the estimation accuracy. \\
\underline{Complexity:} Updating the sketch with an incoming edge takes $O(|L|)$ time, where $|L|$ is the number of distinct labels.\\

\begin{algorithm}[t]
\floatname{algorithm}{\small Algorithm}
\caption{\small UpdateSketchFreqQuery ($sbgSketch,~E<a,~b,~l,~w>$)}
\label{Algo:UpdateEdgeQuerySketch}
\begin{algorithmic}[1]
{\small
\STATE $h(a) \gets H_{v}(a)$ //~hash Vertex~$a$
\STATE $h(b) \gets H_{v}(b)$ //~hash Vertex~$b$
\STATE $RV \gets getRankVector(a,~b,~l)$

\FOR {each label identifier $i \in sbgSketch.EdgeLabels$}
	\STATE //~get the cell in Matrix~$M_i$
	\STATE $Cell \gets sbgSketch.getMatrix(i)[h(a)][h(b)]$ 
    \IF {RV[i] is higher than Cell.Rank}
    	\STATE $Cell.Value \gets w$ //~Evict
		\STATE $Cell.Rank \gets RV[i]$ //~Occupy
    \ELSIF {RV[i]~=~Cell.Rank}    
		\STATE $Cell.Value \gets Cell.Value~+~w$ //~Contribute
	\ENDIF  
\ENDFOR
}
\end{algorithmic}
\end{algorithm}

\noindent
{\bf Edge-Query Estimation:}~Algorithm~\ref{Algo:EstimateEdgeQuery} depicts how \ourSketch{} estimates
the answer to an edge-frequency query.
Figure~\ref{Fig:EdgeCountEval} illustrates how Query~$\hat{f}_{e}(a,~b,~L_0)$ is evaluated. First, the endpoint vertexes are hashed to determine the candidate cells to check at each matrix. Then, only the candidate cells with ranks equal to 
%that 
those 
of the queried edge (i.e., $(a,~b,~L_0)$), are considered by computing the minimum values of these cells (Lines~$8-9$ of Algorithm~\ref{Algo:EstimateEdgeQuery}).
This guarantees that the estimate might be an overestimate of the actual answer, but can never be an underestimate (as each edge is guaranteed to have the highest rank in one matrix). 
If multiple sketches are used, then the minimum value of the results from all the sketches will form the final answer.
Notice that if anyone of the candidate cells has a rank
higher than the corresponding rank of the edge query, then \ourSketch{} returns zero, indicating with certainty that the edge has never been encountered
before (see Theorem~\ref{Theory:zeroEstimation}).
For the sake of completeness, we provide a theoretical error-estimate of \ourSketch's error distribution. \\
\underline{Complexity:} Approximating the aggregate weight of an edge takes $O(|L|)$ time, where $|L|$ is the number of distinct labels.

\begin{theorem}[Informal]
\label{t:boundEdge}
Let $L$ be the number of priorities and $X_e$ be the number of arrivals of edge $e \in V\times V$ during an observation window, where $V$ is the set of vertexes in the graph stream. 
Let $X_e^\text{(\ourSketch)} \geq X_e$ be the upper bound on $X_e$ given by \ourSketch{} and let $X_e^\text{(TCM)} \geq X_e$ be the absolute-error distribution given by TCM with the same number of sketch counter cells. 
Let $P \geq 1$ be the number of $P$-independent hash functions used in \ourSketch{} and TCM.
Then, the distribution of the absolute error is
\begin{align*}
&\text{Pr}[X_e^\text{(\ourSketch)} - X_e > k]  < (\text{Pr}[X_e^\text{(TCM)} - X_e > k ] - \zeta_{k,L,p})^P, 
\end{align*}
where $0 < \zeta_{k,L,P} < \text{Pr}[X_e^\text{(TCM)} - X_e > k ]$ is a lower bound on the probability that $X_e$ is hashed into one of the lower-priority sketch-matrices and survives eviction ($\zeta_{k,L,P}$ is given in Appendix~\ref{sec:Proof_BoundEdge}).
\end{theorem}

%\noindent

\begin{figure}[h!!!]
\centering
\includegraphics[height=1.5in,width=3in]{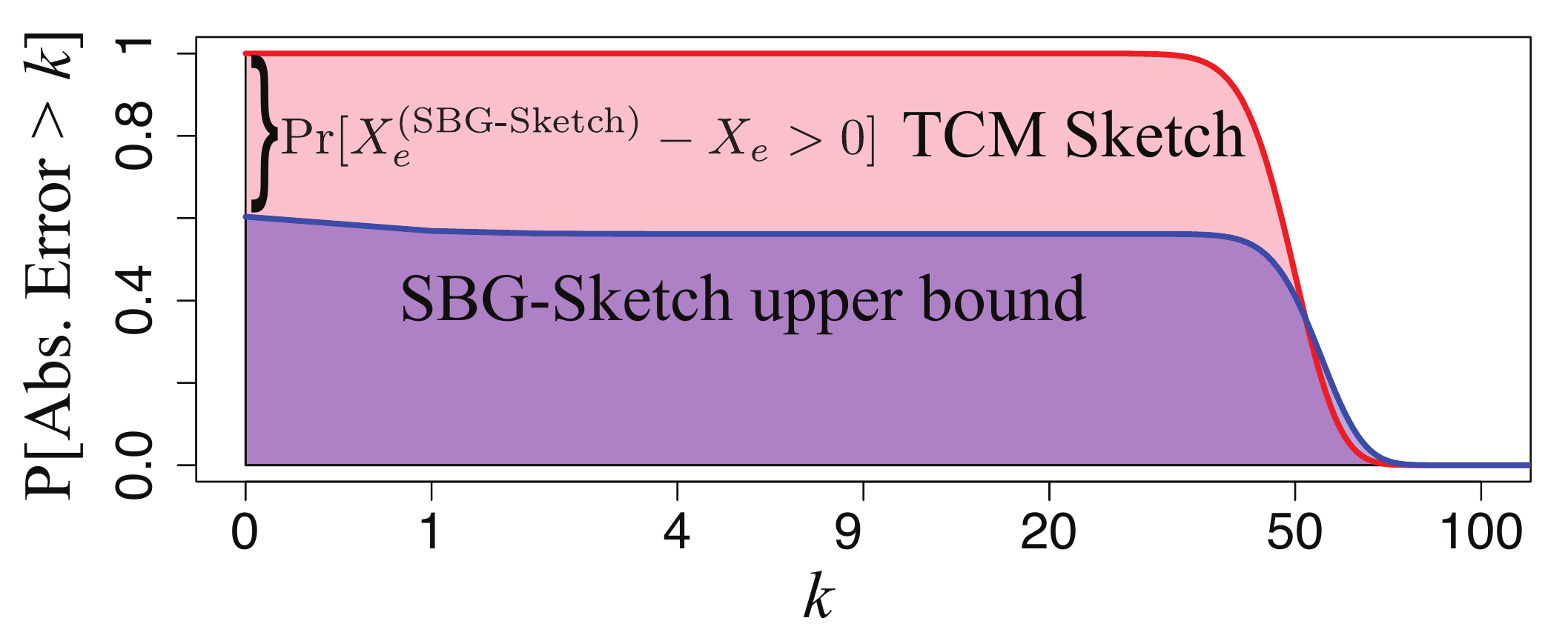}
\caption{SBG-Sketch $\text{Pr}[X_e^\text{(\ourSketch)} - X_e > k]$ upper bound against TCM's exact value $\text{Pr}[X_e^\text{(TCM)} - X_e > k ]$ for highly unbalanced edge arrival rates between different edge labels.}\label{f:boundexample}
\end{figure}

A formal version of Theorem~\ref{t:boundEdge} and its proof are left to Appendix~\ref{sec:Proof_BoundEdge}.
Theorem~\ref{t:boundEdge} shows that, for the same number of sketch cells, the absolute error of \ourSketch{} is smaller than that of TCM.
For values of $L \geq 3$, the value of $\zeta_{k,L,P}$ tends to increase quickly with $k$ until it reaches the probability that the counter related to a given label is evicted in the sketch matrices of other labels.
As an illustrative example, we use the equations in Theorem~\ref{t:boundEdge} to plot the curves in Figure~\ref{f:boundexample}. Figure~\ref{f:boundexample} 
%shows 
gives 
the complementary cumulative distribution of the absolute error of edges of the most frequent label, say Label $A$, out of 100 distinct labels in \ourSketch{} against that of TCM when taking into account a 10\% decrease in the sketch matrix size due to the ranking data structure. We set the edge arrival-rate of Label~$A$ to be such that there is an average of $50$ edge collisions per sketch counter, which we define as $100\times$ that of other $99$ labels; we consider only one hash function for simplicity ($P=1$). 
Note that \ourSketch{} gets absolute error of {\bf zero} with probability $1 - \text{Pr}[X_e^\text{(\ourSketch)} - X_e > 0] \approx 0.4$ while TCM with the same probability gets absolute error around 40. This happens because if \ourSketch{} is not able to have the Label~$A$ counter evicted from the sketches of other labels, it will very likely contain the correct number of edge arrivals of edges of Label~$A$.

\begin{theorem}
\label{Theory:zeroEstimation}
{\em Using \ourSketch{}, ~$\hat{f}_{e}(a,~b,~L_i)~=~0$ $\implies$ $f_{e}(a,~b,~L_i)~=~0$.}
%{\em ~$\hat{f}_{e}(a,~b,~L_i)$ is evaluated as zero by \ourSketch{} if ~$f_{e}(a,~b,~L_i)~=~0$.}
\end{theorem}

\noindent
{\bf Proof}.~Refer to Appendix~\ref{sec:Proof_ZeroEstimate}.

\begin{algorithm}[t]
\floatname{algorithm}{\small  Algorithm}
\caption{\small EstimateEdgeQuery ($sbgSketch,~E<a,~b,~l>$)}
\label{Algo:EstimateEdgeQuery}
\begin{algorithmic}[1]
{\small
\STATE $freqEstimate \gets \infty$ //~for the get-min logic
\STATE $h(a) \gets H_{v}(a)$ //~hash Vertex~$a$
\STATE $h(b) \gets H_{v}(b)$ //~hash Vertex~$b$
\STATE $RV \gets getRankVector(a,~b,~l)$

\FOR {each label identifier $i \in sbgSketch.EdgeLabels$}
	\STATE $Cell \gets sbgSketch.getMatrix(i)[h(a)][h(b)]$ 
    \STATE //~only check cells the edge may have contributed to their values
    \IF {RV[i]~=~Cell.Rank}     
		\STATE $freqEstimate \gets min(freqEstimate,~Cell.Value)$
    \ELSIF {RV[i] is higher than Cell.Rank}  
    	\STATE //~this edge was never seen before
		\STATE $freqEstimate \gets 0$ 
        \STATE $break$ 
	\ENDIF 
\ENDFOR
\IF {freqEstimate~=~$\infty$}     
		\STATE $freqEstimate \gets 0$
\ENDIF 
\RETURN freqEstimate
}
\end{algorithmic}
\end{algorithm}

\begin{figure}[t]
\centering
   \subfigure[Before Processing Edge~$E$]{\includegraphics[width=3.3in]{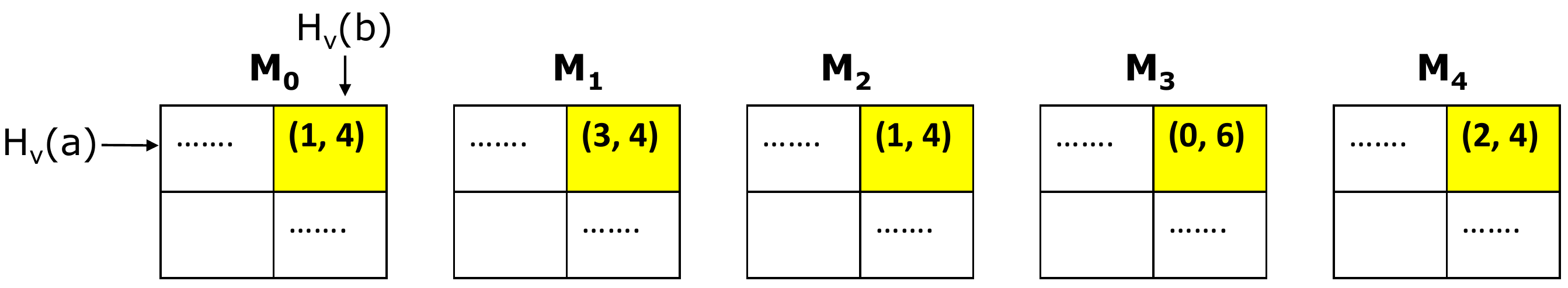}
	\label{Fig:EdgeCountBefore}
	}        
  %  \hspace{0.01in}
    \subfigure[After Processing Edge~$E$]{\includegraphics[width=3.3in]{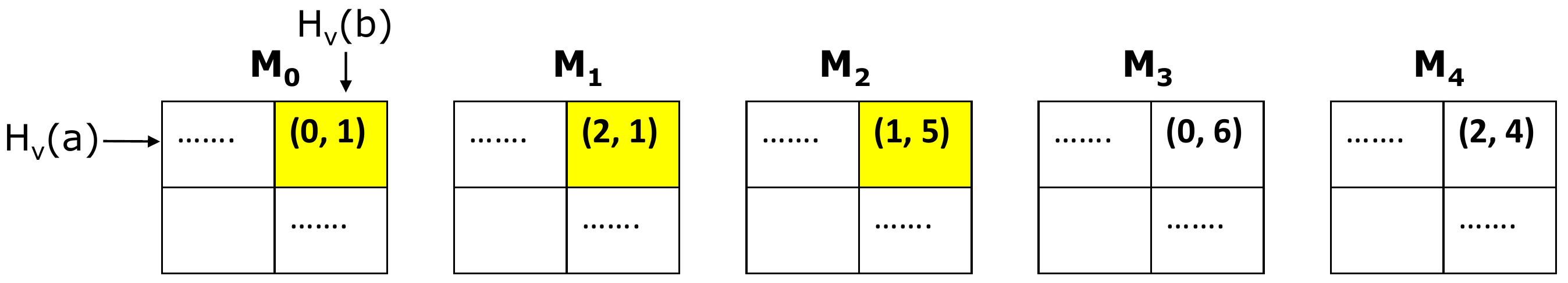}
	\label{Fig:EdgeCountAfter}
	}        
\caption{\ourSketch{} before and after streamed edge~$E$ = (a, b, $L_0$, 1) with rank vector of [0, 2, 1, 4, 3]}
\label{Fig:EdgeCountStreaming}
\end{figure}

\begin{figure}[t]
\centering
\includegraphics[width=3.3in]
{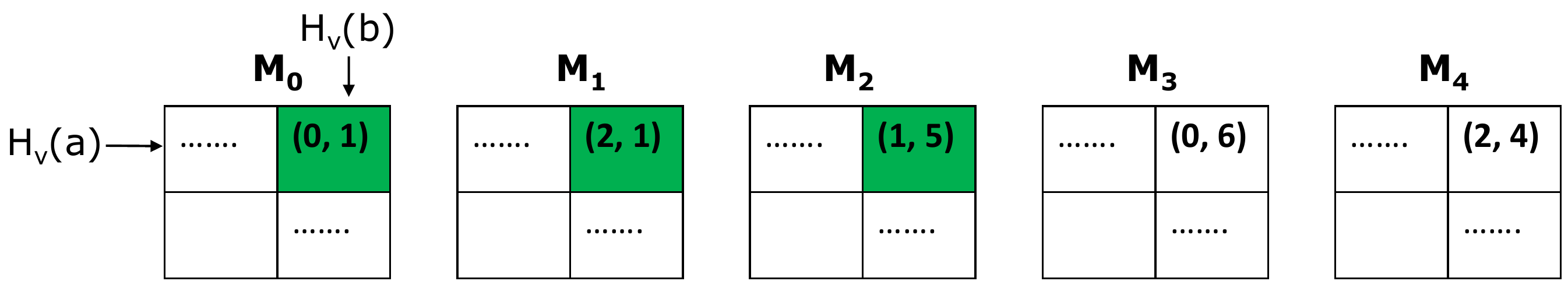}
\caption{Estimating the frequency of Edge~$E$ = (a, b, $L_0$) with rank vector [0, 2, 1, 4, 3].}
\label{Fig:EdgeCountEval}
\end{figure}

\subsubsection{Sub-Graph Queries}
\label{sec:subGraphQueries}
Aggregating edge-weights of a sub-graph is considered in both gSketch~\cite{gSketch_VLDB_2011} and TCM~\cite{TCM_SIGMOD_2016}. However, \ourSketch{} expands the semantics of sub-graph aggregate queries to allow restricting the sub-graph query by the edge labels. Given a sub-graph, say $g$, identified by a set of labeled edges, say $Q~=~\{(a_1,b_1,l_1),~\ldots,~(a_s,b_s,l_s))\}$, an exact sub-graph aggregation query $f_g(Q)$ returns the minimum of the weights of all the edges listed by~$Q$. We denote the estimate of a sub-graph aggregation query by $\hat{f}_g(Q)$, and we adopt the semantics that if the estimated frequency of any edge in~$Q$ is~$0$, we estimate $\hat{f}_g(Q)$ to be~$0$.
The reason is that 
the sub-graph identified by~$Q$ does not have an exact match in the stream.
Notice that inserting edges in a sketch that supports sub-graph queries follows the exact logic of Algorithm~\ref{Algo:UpdateEdgeQuerySketch}. Also, evaluating a sub-graph query depends on evaluating the edge-weight estimate of each edge forming the sub-graph (by directly referencing Algorithm~\ref{Algo:EstimateEdgeQuery}).

\underline{Complexity:} Approximating a sub-graph query, say $Q$, takes $O(|L| * |Q_E|)$ time, where $|L|$ is the number of distinct labels, and $|Q_E|$ is the number of edges of Query $Q$.\\

\subsection{Graph Traversal Queries}
\label{sec:traversalQueries}
Graph traversal queries on labeled graphs arise in many application domains. For example, reachability queries are used in communication-network troubleshooting, and random walking on edge-labeled graphs is a primitive operation in many machine learning techniques (e.g., ~\cite{RandomWalk_ML_2010,RandomWalk_CL_2011,Sun:2012}). In this section, we demonstrate how \ourSketch{} is useful in estimating constrained-reachability queries.

Given two vertexes, say $a$ and $b$, and a set of allowed labels, say $lSet$, a constrained reachability query $f_r(a,~b,~lSet)$ returns true if and only if there is a path, say $P$, from Vertex~$a$ to Vertex~$b$ such that each edge of Path~$P$ is labeled by any of the $lSet$'s labels. We denote the estimate of a reachability Query $f_r(a,~b,~lSet)$ as $\hat{f}_r(a,~b,~lSet)$.

Constrained reachability queries are important primitive operations in many application domains. For instance, in a protein-interaction network, where a vertex represents a protein and a labeled edge represents an interaction type between two proteins, a user may need to estimate if two proteins interact directly or indirectly through covalent or stable interaction types, i.e., evaluating $\hat{f}_r(Protien_a,~Protien_b,~\{Covalent, Stable\})$.
In machine-learning applications, one can use constrained-reachability queries as a way to construct feature vectors, indicating whether or not Vertex $A$ can reach Vertex $B$ using only edges of certain labels. These features can be used for link prediction tasks (similar to Sun et al.~\cite{Sun:2012}), among other applications where the learning algorithm can tolerate reachability approximation (i.e., false positives).

\noindent
{\bf Insertion of Edges:}~To support reachability queries, the same logic used to insert edges for edge queries could be applied (see Algorithm~\ref{Algo:UpdateEdgeQuerySketch}). However, it is sufficient to use edge weights of one to indicate edge existence between two vertexes, where an edge weight of zero in the adjacency matrices of the sketch flags that no edge exists.

\noindent
{\bf Reachability-Query Estimation:}~Any traversal-based reachability algorithm (e.g., DFS search) can be used to traverse the adjacency matrices of \ourSketch{} to evaluate a constrained-reachability query. However, the algorithm should check only the edges labeled by any of the labels allowed by the query. Notice that if multiple sketches are used, then each sketch evaluates the query independently. Then, the independent results are {\em logically anded} %together 
to form the final answer. The time-complexity of the evaluation is determined by the third-party algorithm used in traversing the summarized topology of the sketch.\\

%\section{Optimizing for Accuracy}
%\label{sec:accuracyOptimization}

\section{Experimental Evaluation}
\label{sec:ExperimentalEvaluation}

In this section, we experimentally evaluate the accuracy and the performance of \ourSketch{} against TCM~\cite{TCM_SIGMOD_2016}, the only state-of-the-art that is comparable to \ourSketch{} w.r.t. query expressiveness.
We measure the processing time and the estimation error using 
various 
%different 
types of queries on real datasets from different domains. For measuring the estimation errors, we use the {\em average relative-error} metric (ARE, for short). As defined in~\cite{gSketch_VLDB_2011,TCM_SIGMOD_2016}, given Query~$Q_i$, the relative error of $Q$, say $re(Q_i)$, is defined as:
\begin{align*}
re(Q_i)~=~\frac{\hat{f}(Q_i)~-~f(Q_i)}{f(Q_i)}
\end{align*}
where $f(Q_i)$ is the actual result of Query~$Q_i$, and $\hat{f}(Q_i)$ is its estimated value. The average relative-error is computed over a set of queries, say $Q = {Q_1, Q_2, \ldots, Q_n}$, as follows:
\begin{align*}
ARE(Q)~=~\frac{\sum_{j=1}^{n} re(Q_j)}{n}
\end{align*}

\subsection{Datasets and the Experimental Setup}
\label{sec:Datasets}
We use real datasets of labeled graphs from three different domains (communication network, social network, and biological network).
We use IPFlow~\cite{IPFlow}, Youtube~\cite{Youtube}, and String~\cite{String}. Table~\ref{table:Datasets} summarizes the properties of the aforementioned datasets, and gives the number of distinct labels of each dataset. To verify the label-skewness in real datasets, we found that for all the datasets in Table~\ref{table:Datasets}, $65\%-92\%$ of the edges are labeled by only $10\%-24\%$ of the labels (i.e., frequent labels). The IPFlow dataset is a collection of anonymized communication-traces from CAIDA’s equinix-Chicago monitor, where the edge labels represent the communication protocol used (e.g., HTTPS, SMTP, Telnet). The Youtube dataset is a subset of the popular video-sharing service, where the vertexes represent users, the edges represent user interactions, and the edges are labeled by user-interaction types as described in~\cite{Youtube}. The String dataset is a protein-interaction network, where the vertexes represent proteins, the edges represent interactions among the proteins, and the labels represent the protein-protein interaction types. Notice that the number of labels in all the real labeled datasets is much less than the number of the vertexes and the edges.

Both 
\ourSketch{} and TCM~\cite{TCM_SIGMOD_2016} are implemented as C++ libraries that can be used as components by any server. Our experiments are conducted on a machine running Windows~10 on~4 cores of Intel i7 3.40~GHz and 16~GB of main-memory. Notice that TCM is not designed to deal with labeled graphs, however, we follow the suggestion of the original paper~\cite{TCM_SIGMOD_2016} by creating a matrix for each label. Hence, \ourSketch becomes TCM if the ranking logic %illustrated in Section~\ref{sec:RankingLogic}
is removed. For fairness,
%comparison, 
we use the same memory sizes for both \ourSketch{} and TCM.
Also note that if the rank logic data structures of \ourSketch{} do not reduce the number of sketch counters, the accuracy of \ourSketch{} is lower bounded by the accuracy of TCM. The reason is that an edge in \ourSketch{} is given the highest priority in the matrix corresponding to its label. Hence, it is guaranteed that an edge will experience the same hash-collision rate in the matrix corresponding to its label in both \ourSketch{} and TCM. For this reason, we focus on queries constrained with labels of high-frequency as they show the power of \ourSketch{} to leverage unused cells in the matrices corresponding to less-frequent labels.

\begin{table}
%\centering
\begin{tabular}{|l||p{1.5cm}|p{1.5cm}|p{1.30cm}|} \hline
Dataset&\# Vertexes&\# Edges&\# Labels\\ \hline
IP Flow&237,022&22,497,005&39\\ \hline
Youtube&15,088&13,628,895&5\\ \hline
String Protein Network&1,520,673&348,473,440&45\\ \hline
\end{tabular}
\caption{Datasets}
\label{table:Datasets}
\end{table}

%\vfill\eject
%\pagebreak
\subsection{Constrained Edge-Queries}
\label{sec:ExpEdgeQueries}

\subsubsection{Varying the Sketch Size}
\label{sec:ExpEdgeQueries_VarySize}
In this set of experiments, we study the accuracy of approximating constrained edge-queries using \ourSketch{} and TCM, the state-of-the-art. We fix the number of hash functions to two (i.e., we set $P~=~2$ in Figure~\ref{Fig:GeneralSketchStructure}), and we measure the estimation accuracy for 
%different 
various 
sketch sizes. A sketch size is determined using a sketch-size factor. A sketch-size factor, say $F$, is a value between $0$ and $1$ exclusive that defines the memory-size of the sketch w.r.t. the size of the original graph dataset. For instance, if the size of the original graph dataset is $1000$~MBs, and the sketch-size factor is $0.1$, then the total memory-size of a single sketch will be upper bounded by $1000~*~0.1~=100$~MBs. We consider this for each experiment so that both \ourSketch{} and TCM are assigned the same memory size for fair comparisons.

We generate 10,000 constrained edge-queries, say $Q_{10k}$, randomly for each dataset, say $G$, in Table~\ref{table:Datasets}. Then, we run Query-set $Q_{10k}$ on \ourSketch{} and TCM with different sketch-size factors, specifically, $0.05$, $0.1$, $0.15$, $0.2$, $0.25$, $0.3$, and $0.35$. We stream all the edges of each dataset in Table~\ref{table:Datasets} before running the query sets (e.g., for the String dataset, the sketch receives 348 million edges and then we run the 10,000 queries). Figure~\ref{Fig:Edge_Queries} gives the average relative-error when running $Q_{10k}$ as formerly described using each dataset of Table~\ref{table:Datasets}. As expected, the average relative-error decreases when increasing the sketch size for both \ourSketch{} and TCM. 
%This is because 
The reason is that the number of 
collisions decrease as the sketch size increases, and the average relative-error decreases accordingly. For all the datasets, the ARE 
%due to 
of
\ourSketch{} is less than that of TCM. We attribute this to the rank-vectors and the ranking logic of \ourSketch{}. 
%This is because 
The reason is that 
the rank-vectors and the ranking logic are the main differences between \ourSketch{} and TCM (i.e., removing the ranking logic turns \ourSketch{} to TCM). Notice that for the same memory-size, the number of cells allocated to TCM is higher than that allocated to \ourSketch{} (a cell in \ourSketch{} uses an additional byte for the rank). Although the number of cells in TCM is higher, \ourSketch{} achieves better average relative-error as the ranking logic automatically handles label skew, and leverages the cells that may not be used by low-frequency labels. In contrast, an edge in TCM of Label~$L_i$ assigned to Matrix~$M_i$ can never use a cell in another matrix, say Matrix~$M_j$, even if $M_j$ is not fully-occupied by edges of Label~$L_j$.

Notice that the accuracy of \ourSketch relative to that of TCM increases as the graph size increases (which is used also to define the sketch size). To 
%show this, 
illustrate,
we measure the TCM error that \ourSketch{} reduces (e.g., a $90$\% error reduction means that the average relative-error of \ourSketch{} is only $10$\% of the error in TCM). Figure~\ref{Fig:Edge_ErrorReduction} gives the error reduction caused by \ourSketch{} comparing to that of TCM for all the datasets. \ourSketch{} significantly reduces the error of TCM, where the error reduction reaches $99$\% for the large String dataset, i.e., the estimation error of \ourSketch{} is just $1$\% of the TCM error on the same setup. Notice that the error reduction increases as the graph size increases. For example, the reduction reaches $88$\% for the Youtube dataset, where the size of the Youtube dataset is relatively smaller than that of the IPFlow dataset (whose the error reduction reaches $97$\%). We attribute this to the cell utilization effectiveness of the ranking module of \ourSketch{}. In contrast, TCM is vulnerable to wasting more cells if they are assigned to larger matrices of labels that are low-represented by graph edges.

The results in Figure~\ref{Fig:Edge_Queries} illustrate that the accuracy of \ourSketch{} is significantly and consistently better than that of TCM over real data from different domains. For example, consider the String protein-interaction network in Figure~\ref{Fig:Edge_String}. When setting the sketch size to $0.25$ of the String dataset size, the average relative-error of \ourSketch{} is $0.14$, which is significantly better than $11.92$, the average relative-error of TCM for the same setup (i.e., TCM overestimates the queries by $11.92x$ of the actual values on average, while the overestimation by \ourSketch is just $0.14x$). Moreover, the accuracy of \ourSketch{} increases when increasing the number of the pairwise-independent hash functions as demonstrated in Section~\ref{sec:ExpEdgeQueries_VaryHash}.

\begin{figure*}[ht]
\centering
    \subfigure[IPFlow dataset]{\includegraphics[width=2.24in]{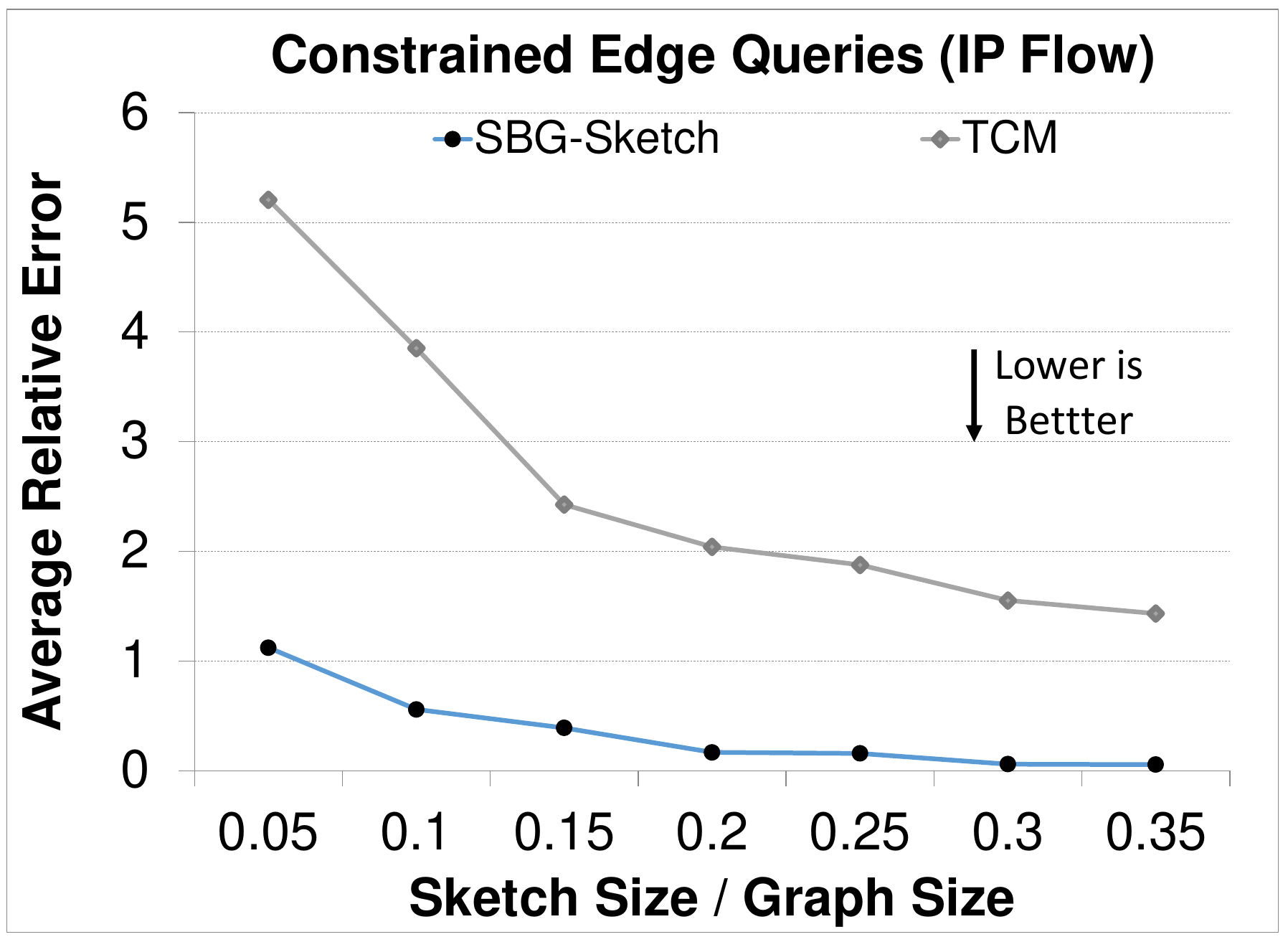}
	\label{Fig:Edge_IPFlow}
	}
   \hspace{0.01in}
\subfigure[Youtube dataset]
{\includegraphics[width=2.24in]{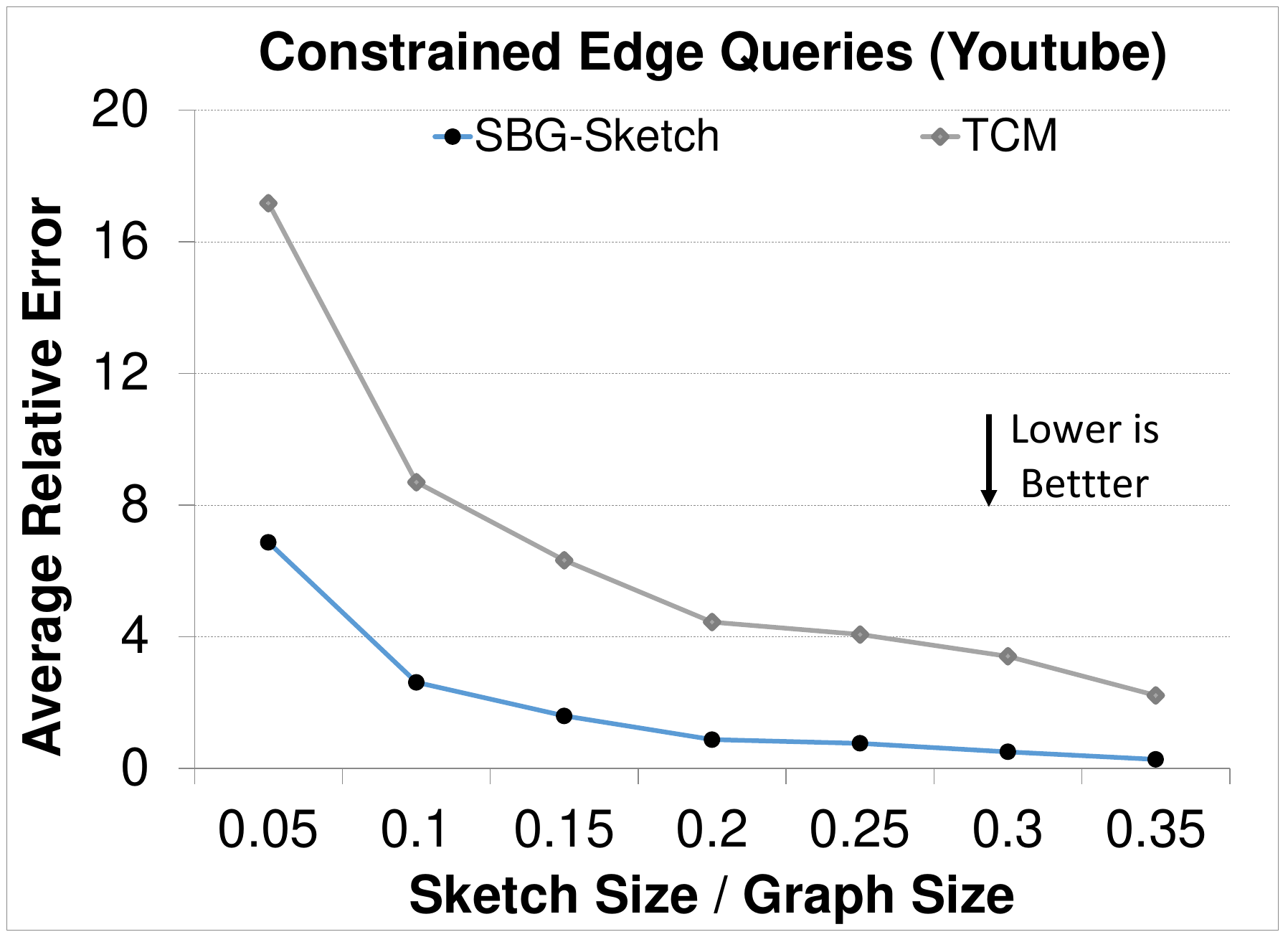}	\label{Fig:Edge_Youtube}}
\hspace{0.01in}
\subfigure[String dataset]
{\includegraphics[width=2.24in]{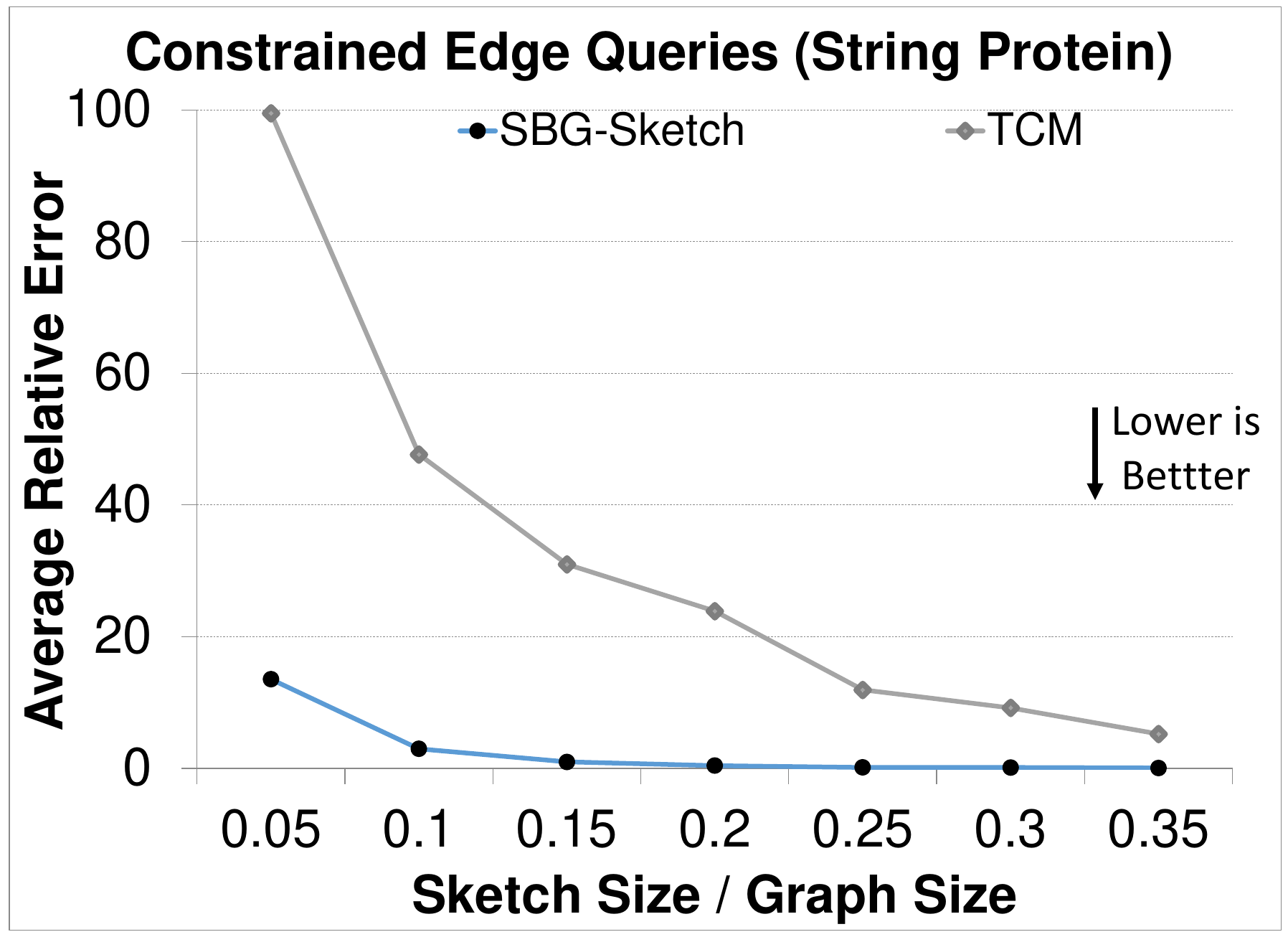}	\label{Fig:Edge_String}}
\caption{\ourSketch reduces the estimation error of TCM by up to 99\% in approximating constrained edge queries.}
\label{Fig:Edge_Queries}
\end{figure*}

\begin{figure}[t]
\centering
\includegraphics[width=2.2in]
{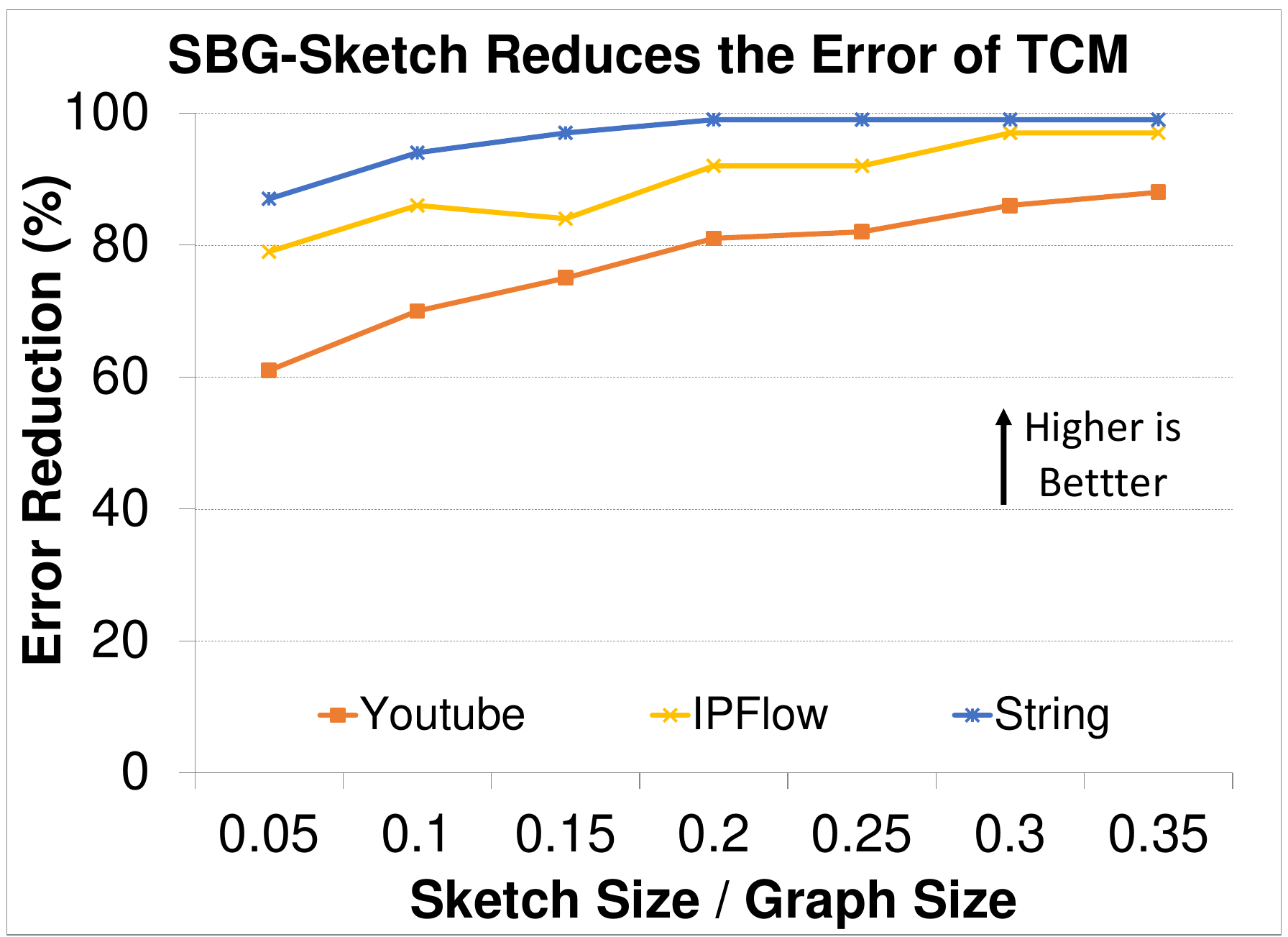}
\caption{\ourSketch{} increases the accuracy significantly of approximating constrained edge-queries.}
\label{Fig:Edge_ErrorReduction}
\end{figure}

\subsubsection{Using Multiple Hash Functions}
\label{sec:ExpEdgeQueries_VaryHash}

In this experiment, we measure the effect of using multiple sketches. Each sketch uses a different hash function to hash the vertexes into rows and columns of its matrices. The hash functions form a set of pairwise-independent hash functions.
In particular, we vary $P$, the number of hash functions (as in Figure~\ref{Fig:GeneralSketchStructure}), while fixing the memory size of each sketch. Hence, the total memory size increases with $P$. There are two reasons behind the setup of this experiment. First, we need to study the effect of increasing the number of hash functions 
%that will be clear when 
while
fixing all the other parameters including the dimensions of each matrix. Second, the setup is consistent with the same experimental setup reported by TCM~\cite{TCM_SIGMOD_2016}. 

In this experiment, we use the same set of queries described in Section~\ref{sec:ExpEdgeQueries}, namely $Q_{10k}$. %Again, 
The $10k$ query set executes after inserting into the sketch the entire datasets described in Table~\ref{table:Datasets}. We fix the size of a single sketch to be $0.1$ of the size of each queried dataset, and we vary the number of the hash functions from $1$ to $7$.  
Figure~\ref{Fig:EdgeMultiHash} gives the average relative-error when running $Q_{10k}$ as formerly described using each dataset of Table~\ref{table:Datasets}. The average relative-error decreases as the number of pairwise-independent hash functions increases for both methods. However, \ourSketch{} consistently provides better accuracy than that of TCM. The reason is that each sketch hashes the same edges differently, and this allows the edges to collide differently in each sketch. At query time, the results from all the sketches are used, and the most accurate one dominates as the final result (as explained in Section~\ref{sec:queryEstimation}). Notice that each sketch is updated and is queried independently. The advantage of processing the sketches independently is twofold. First, the accuracy is enhanced as each sketch summarizes the graph stream differently. Second, updating and querying the sketches can be 
%done 
preformed 
in parallel, which allows performance gains.

\begin{figure*}[ht]
\centering
    \subfigure[IPFlow dataset]{\includegraphics[width=2.24in]{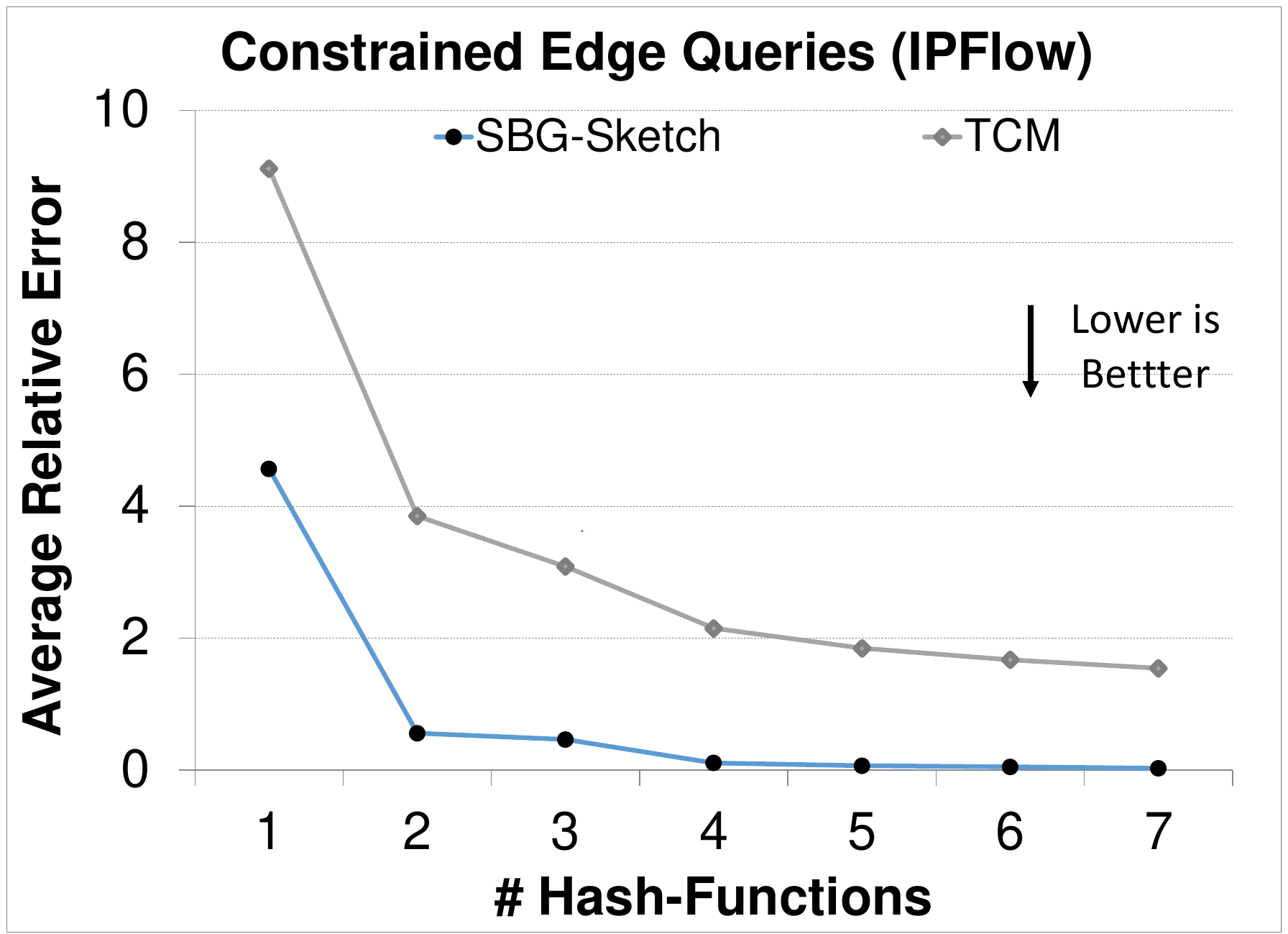}
	\label{Fig:Edge_MultiHash_IPFlow}
	}
   \hspace{0.01in}
\subfigure[Youtube dataset]
{\includegraphics[width=2.24in]{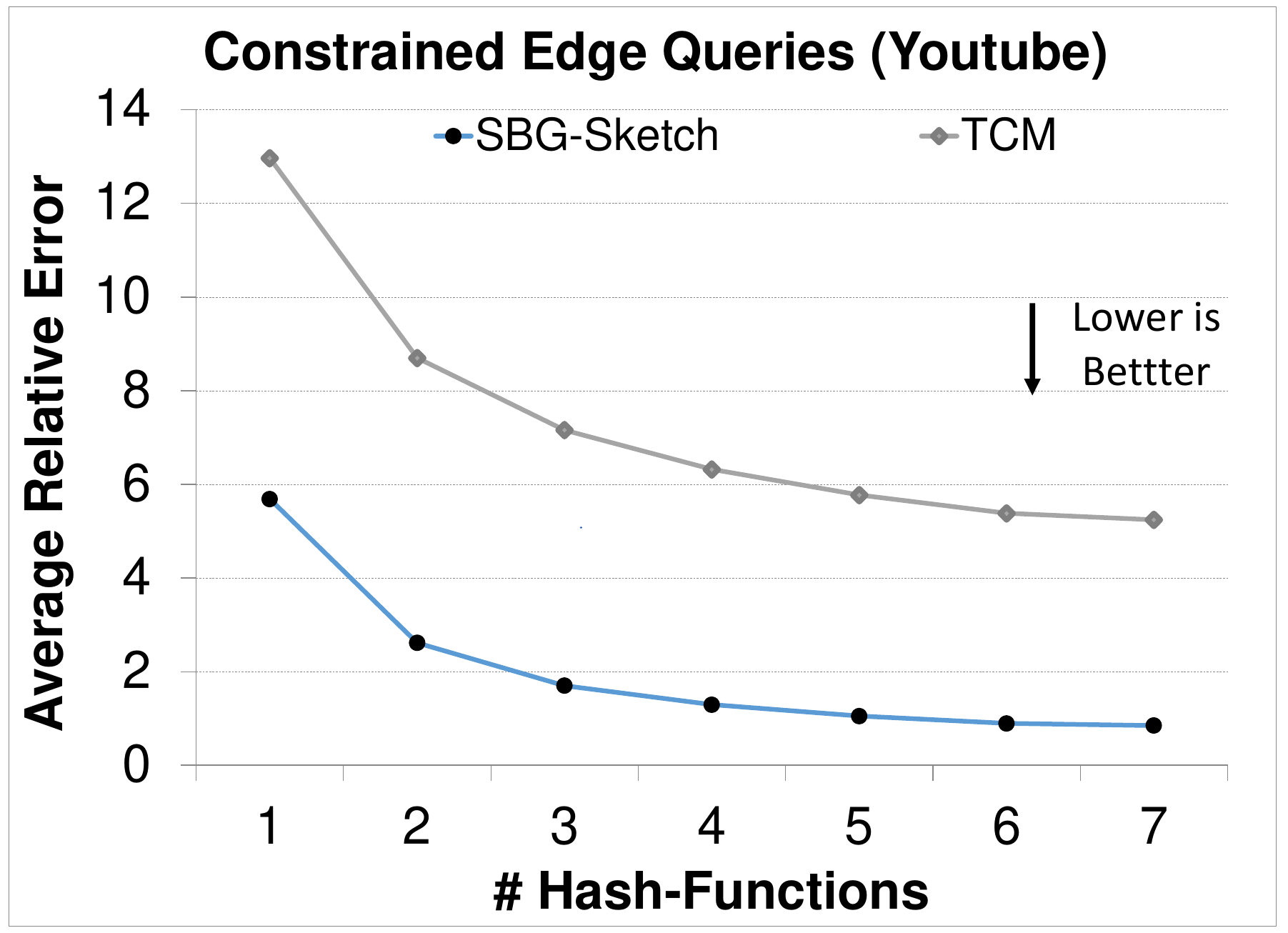}	\label{Fig:Edge_MultiHash_Youtube}}
\hspace{0.01in}
\subfigure[String dataset]
{\includegraphics[width=2.24in]{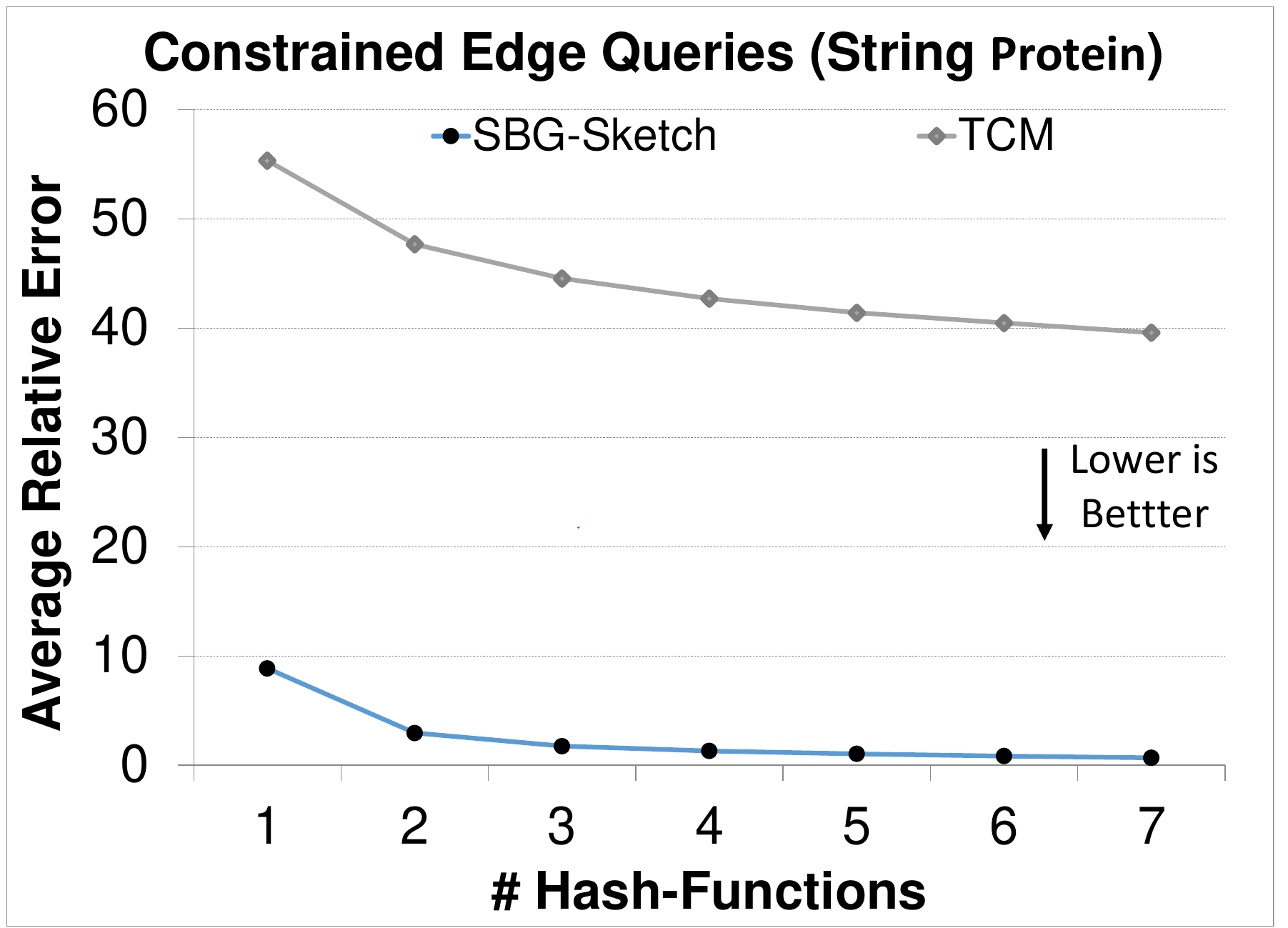}	\label{Fig:Edge_MultiHash_String}}
\caption{The estimation accuracy increases as the number of hash functions increases.}
\label{Fig:EdgeMultiHash}
\end{figure*}

%\begin{figure}[t]
%\centering
%\includegraphics[width=2.3in]
%{Figs/Edge_MultiHash.pdf}
%\caption{The estimation accuracy decreases as the number of hash functions increases while fixing the total memory size.}
%\label{Fig:EdgeMultiHash}
%\end{figure}

\subsection{Constrained Sub-Graph Queries}
\label{sec:ExpSubGraphQueries}

In this set of experiments, we use all the datasets in Table~\ref{table:Datasets} to evaluate the accuracy of estimating constrained sub-graph queries. We generate $10,000$ constrained sub-graph queries randomly with different variations (triangle queries, paths of different lengths, and connected sub-graphs). We measure the accuracy of \ourSketch{} w.r.t. TCM for various sketch sizes while fixing the number of hash functions to two (i.e., $P$ = 2). It is expected to get results that comply with the results of the edge-queries  in Figure~\ref{Fig:Edge_Queries} as the edge query logic is used as a primitive to evaluate sub-graph queries. This set of experiments confirms this expectation as illustrated in Figure~\ref{Fig:Subgraph_Queries}. However, 
%we find that 
the average relative-error reduces for both \ourSketch{} and TCM %comparing 
in contrast to edge queries. 
We attribute this reduction to the conceptual evaluation of the sub-graph queries (see Section~\ref{sec:subGraphQueries}).
In particular, the query result is dominated by the query edge of least frequency. 
Hence, the relative error on average decreases
%comparing 
in contrast to the error 
%of
in 
estimating individual edge queries. Notice that \ourSketch{} always reduces the estimation error %
%of 
over 
TCM (refer to 
%. We attribute this to the reasons we discussed in 
Section~\ref{sec:ExpEdgeQueries}), where the ranking logic of \ourSketch{} leverages more cells than TCM to reduce collisions 
%with the existence 
in the presence 
of skewed-label distributions. Notice that \ourSketch{} handles 
%the label-skewness 
without any %pre-knowledge 
advance knowledge of 
%about the 
label distribution.

\begin{figure*}[ht]
\centering        
    \subfigure[IPFlow dataset]{\includegraphics[width=2.24in]{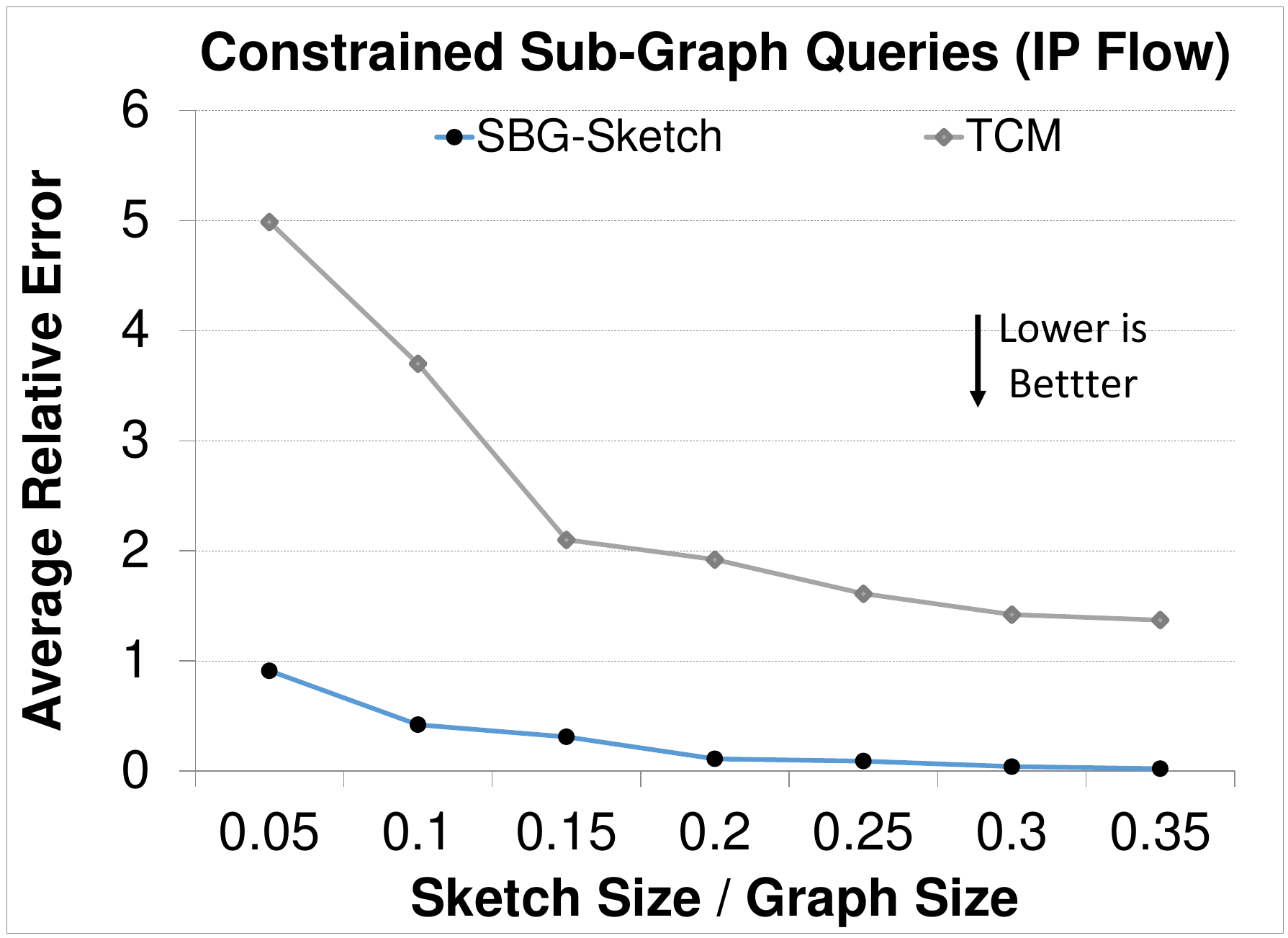}
	\label{Fig:Subgraph_IPFlow}
	}
    \hspace{0.01in}
    \subfigure[Youtube dataset]{\includegraphics[width=2.24in]{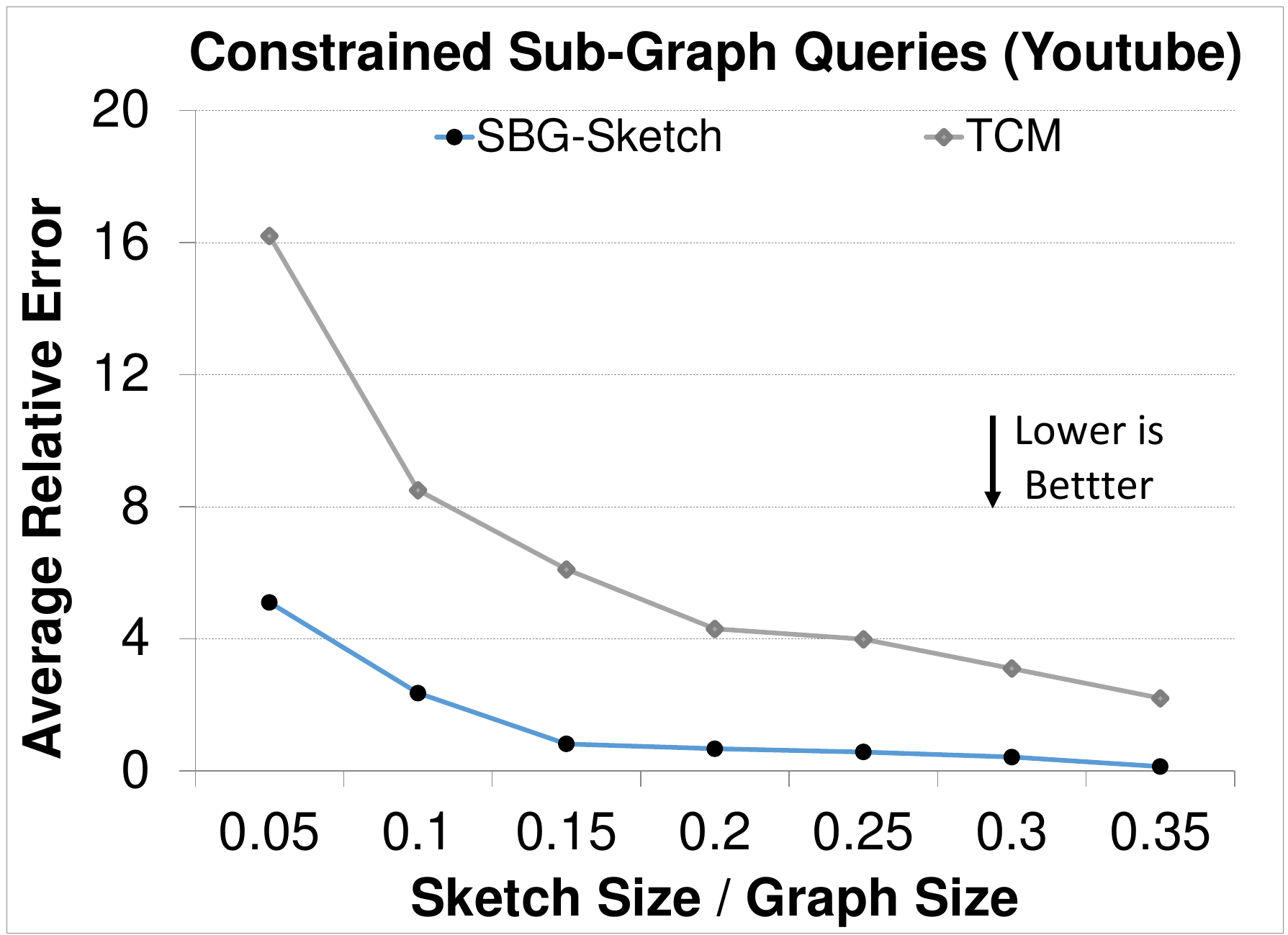}
	\label{Fig:Subgraph_Youtube}
	}
    \hspace{0.01in}
    \subfigure[String dataset]{\includegraphics[width=2.24in]{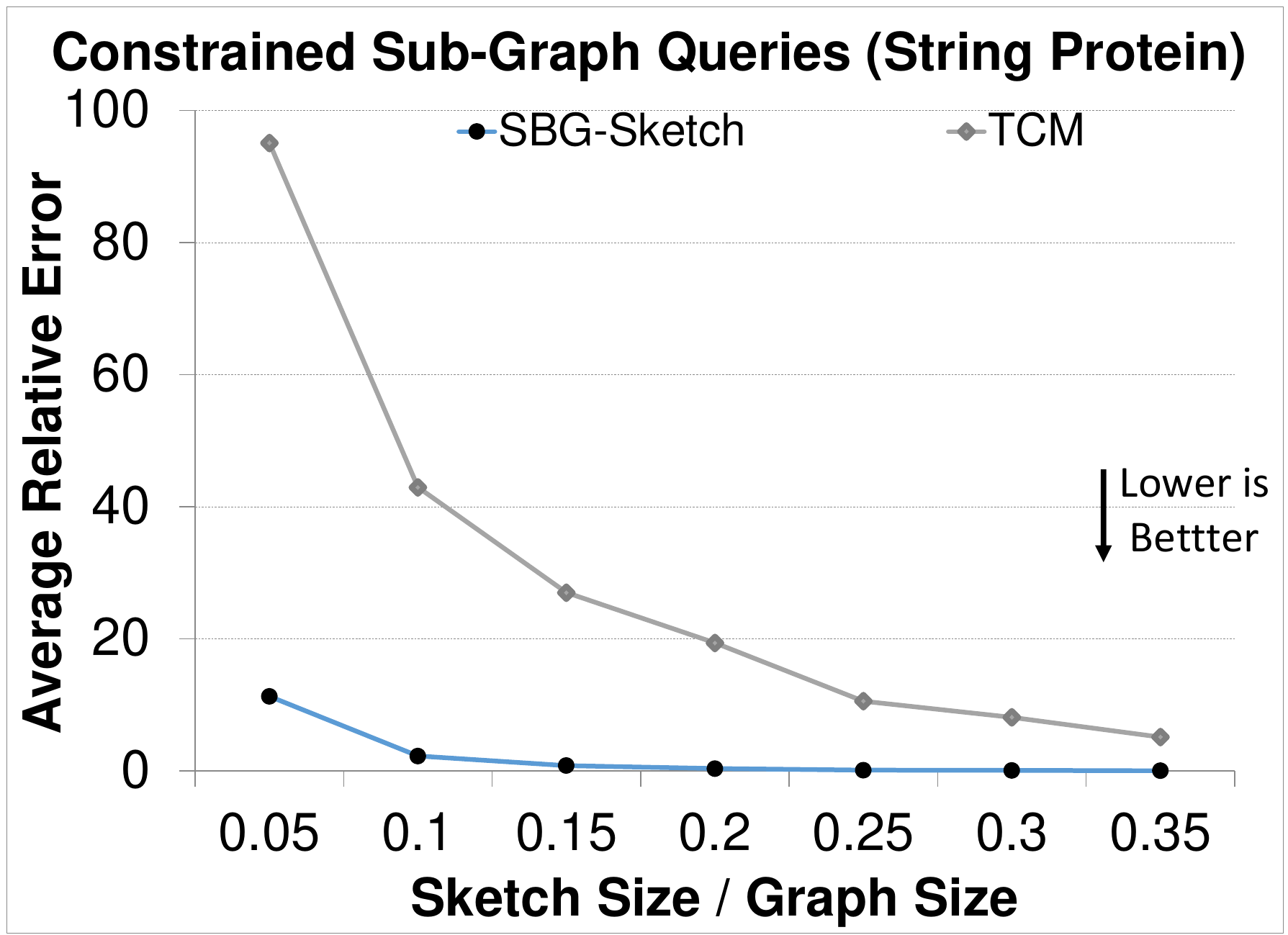}
	\label{Fig:Subgraph_StringDB}
	}
\caption{\ourSketch reduces the estimation error of TCM by up to 84\% in sub-graph query estimation.}
\label{Fig:Subgraph_Queries}
\end{figure*}

%\subsection{Constrained Shortest-Path Lower Bound Queries}
%\label{sec:ExpSPLBQueries}

%\begin{figure}%[t]
%\centering
%\includegraphics[width=2.3in]
%{Figs/SPLB_Tiger.pdf}
%\caption{The average estimate of \ourSketch{} for SPLB queries reaches 323~times that of TCM.}
%\label{Fig:SPLBTiger}
%\end{figure}

\subsection{Constrained Reachability Queries}
\label{sec:ExpReachabilityQueries}

In this set of experiments, we measure the effectiveness of \ourSketch{} to estimate constrained reachability queries. Notice that a reachability query that evaluates to true on the original graph will always evaluate to true using a sketch of that graph. This is true for both \ourSketch{} and TCM as they both keep all the connectives of the input graph streams. However, due to edge collisions, both methods are vulnerable to false positives, i.e., a reachability query that evaluates to false on the original graph might be estimated as true using a sketch of the original graph. Hence, in this set of experiments, we generate random constrained reachability queries with actual results of $false$ on the original graphs (i.e., not reachable), and we measure how many of them are detected as unreachable by both \ourSketch{} and TCM, i.e., we measure the recall of the true-negatives, which is similar to the metric used in~\cite{TCM_SIGMOD_2016} to evaluate the effectiveness of TCM on estimating reachability queries.

We evaluate the true-negative recall of constrained-reachability queries using all the datasets listed in Table~\ref{table:Datasets}. We generate $1000$ random reachability queries, say $Q_{rset}$, where each query is constrained to use up to half the labels of the queried graphs. We ensure that all the generated queries are not reachable in the original graphs. We run Query-set~$Q_{rset}$ with different sketch-size factors on the x-axis of Figure~\ref{Fig:Exp_Reach_Queries} while fixing the number of hash functions to two. The y-axis gives the percentage of the true-negatives recall (the higher the better). 
Figure~\ref{Fig:Exp_Reach_Queries} illustrates that the accuracy of \ourSketch{} in recalling true-negatives is very effective even 
%at 
for 
small sketch sizes.
%, e.g., 
Figure~\ref{Fig:Reach_IPFlow} illustrates that \ourSketch{} and TCM have accuracy of $70.8\%$, and $9.1\%$, respectively, when fixing the sketch size to only $0.05$ of the graph stream size (i.e., \ourSketch{} is $7.8x$ more accurate than TCM).
\ourSketch{} estimates correctly 
over 
%more than 
$90$\% of the queries when setting the sketch size to 
%be 
$0.1$ of the graph size for the IPFlow and the String datasets (see Figure~\ref{Fig:Reach_IPFlow} and Figure~\ref{Fig:Reach_StringDB}), where the accuracy reaches up to $99.6$\%. We attribute this gain in accuracy to the ranking logic of \ourSketch{} that automatically balances %load balances 
the filling of the sketch matrices, and overcomes the issue of skewed labels, and hence, decreasing the overall hash-collision rates.
%due to hashing.

\begin{figure*}[ht]
\centering        
    \subfigure[IPFlow dataset]{\includegraphics[width=2.24in]{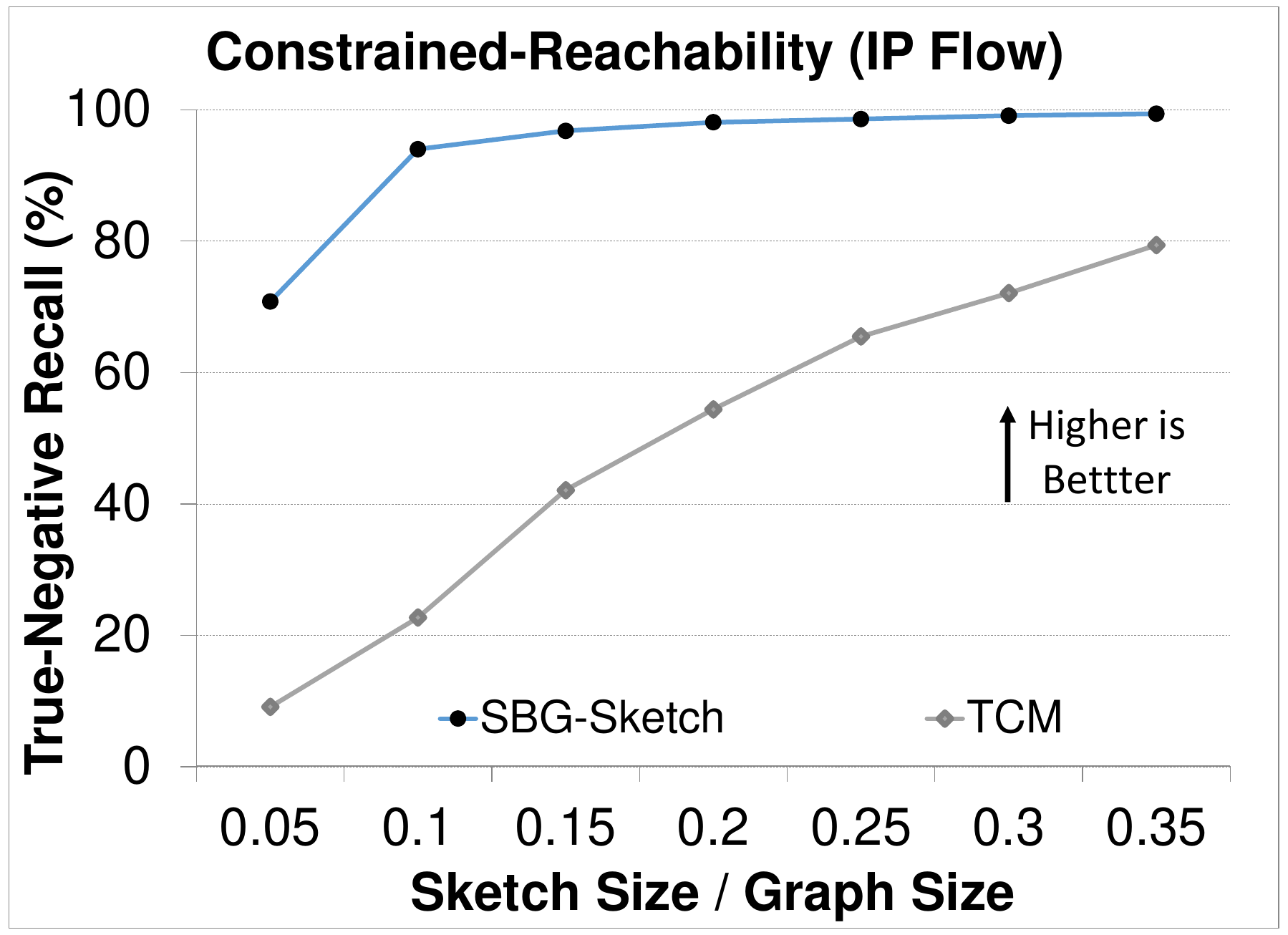}
	\label{Fig:Reach_IPFlow}
	}
    \hspace{0.01in}
    \subfigure[Youtube dataset]{\includegraphics[width=2.24in]{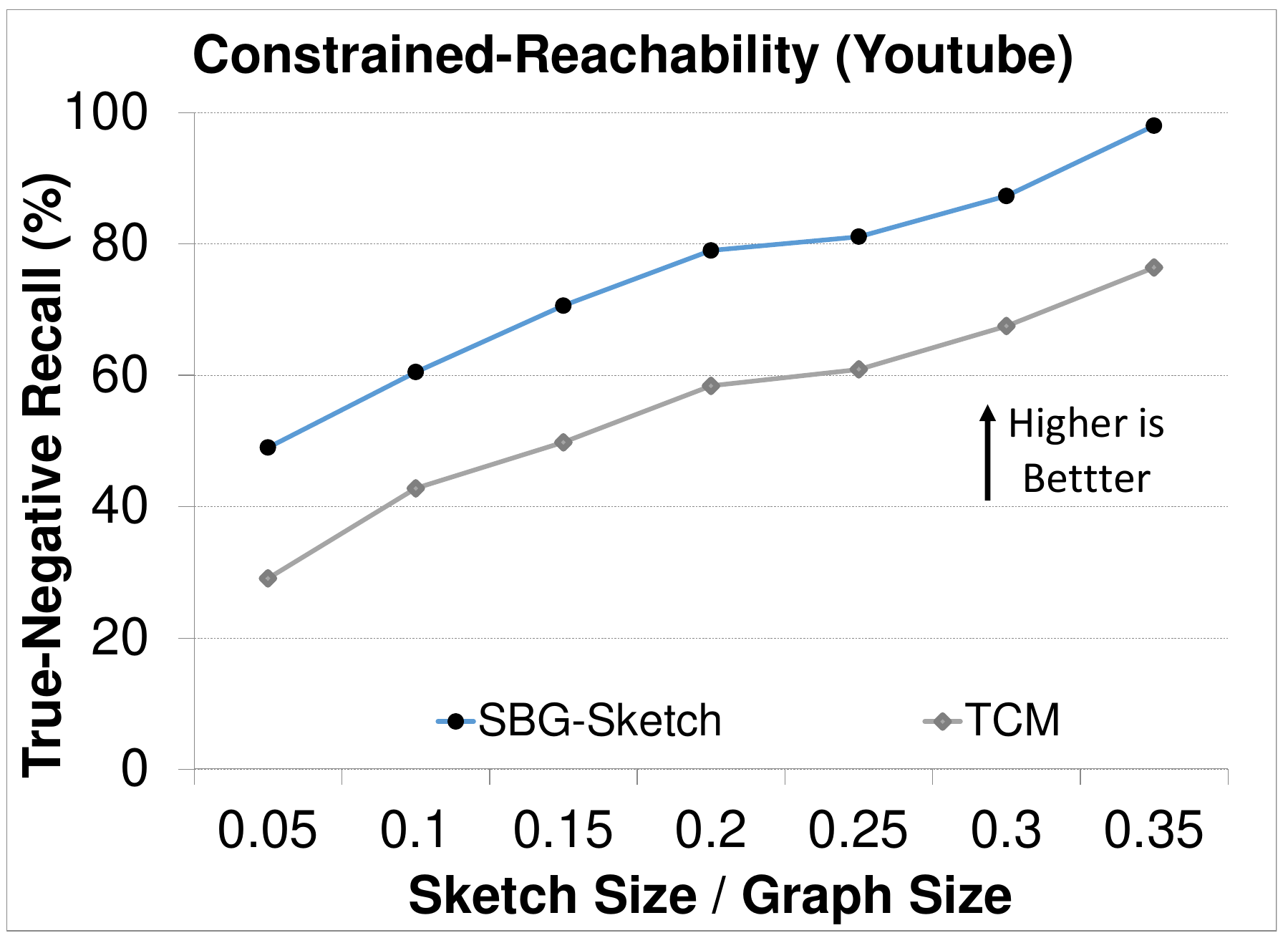}
	\label{Fig:Reach_Youtube}
	}
    \hspace{0.01in}
    \subfigure[String dataset]{\includegraphics[width=2.24in]{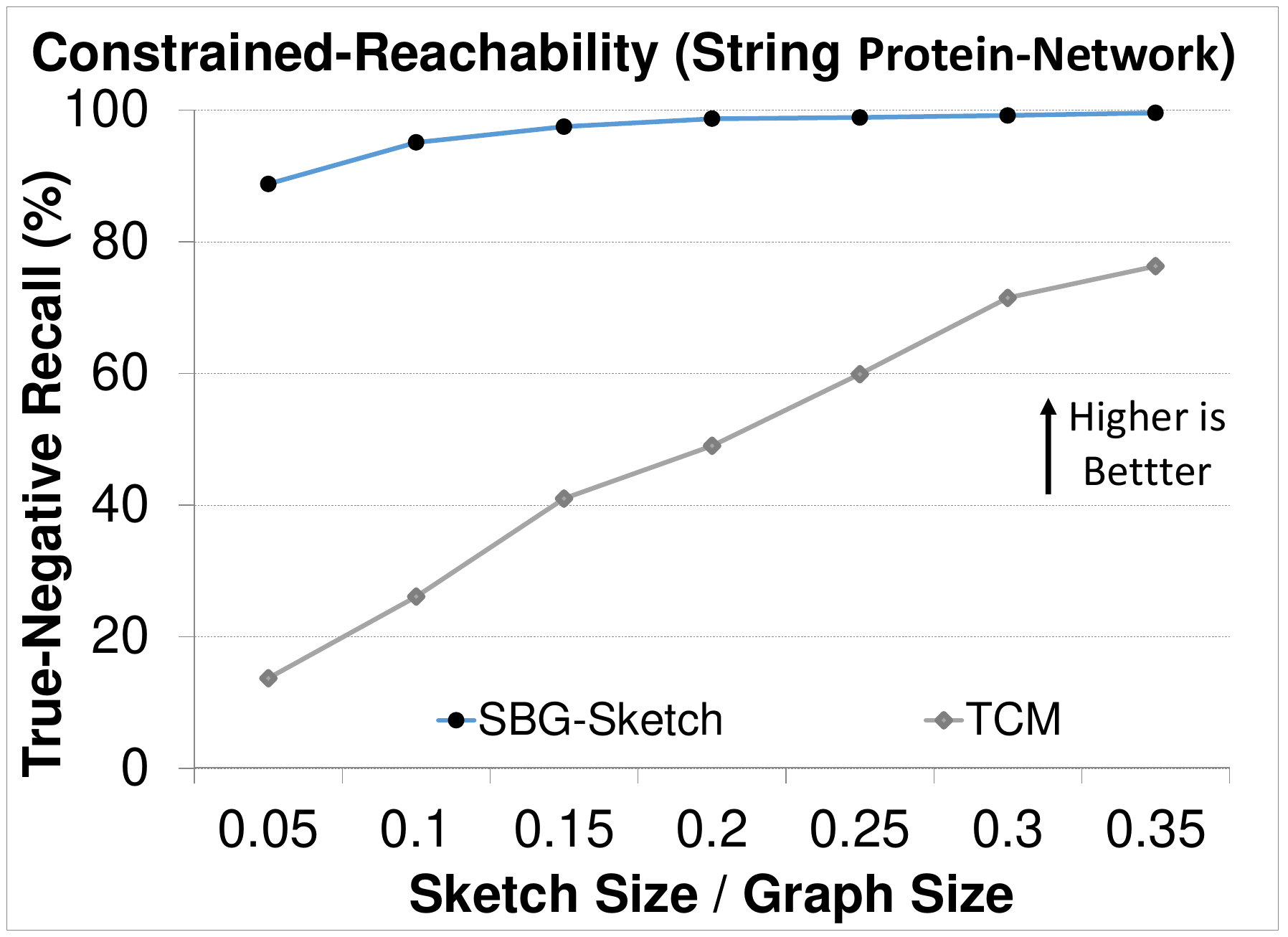}
	\label{Fig:Reach_StringDB}
	}
\caption{\ourSketch{} estimates the true-negative of constrained-reachability queries by up to 99.6\% in accuracy.}
\label{Fig:Exp_Reach_Queries}
\end{figure*}

\subsection{Processing-Time Efficiency}
\label{sec:ExpConstructionTime}

In this set of experiments, we measure the time of constructing \ourSketch{} for each dataset in Table~\ref{table:Datasets}. Notice that inserting edges into \ourSketch{} during  sketch construction  has the same time-complexity as evaluating edge queries, and the same holds for TCM~\cite{TCM_SIGMOD_2016}. As \ourSketch{} performs more logic related to rank-values maintenance, we expect \ourSketch{} to take additional construction time 
in contrast to 
%comparing to 
that of TCM. In this experiment, we fix the sketch-size to be $0.1$ of the graph size and compare the construction time of both \ourSketch{} and TCM. The y-axis in Figure~\ref{Fig:Performance_Construction} gives the construction time in milliseconds for the three datasets listed in Table~\ref{table:Datasets}. Notice that the construction time of \ourSketch{} is comparable to the simpler construction logic of TCM. We observe an average of $28$\% time increase over all datasets. This construction-time increase is 
%may be 
acceptable 
given the significant gain in accuracy in \ourSketch{}.
%to gain significantly-better accuracy when using \ourSketch{}.

\begin{figure}%[H]
\centering
\includegraphics[width=2.2in]
{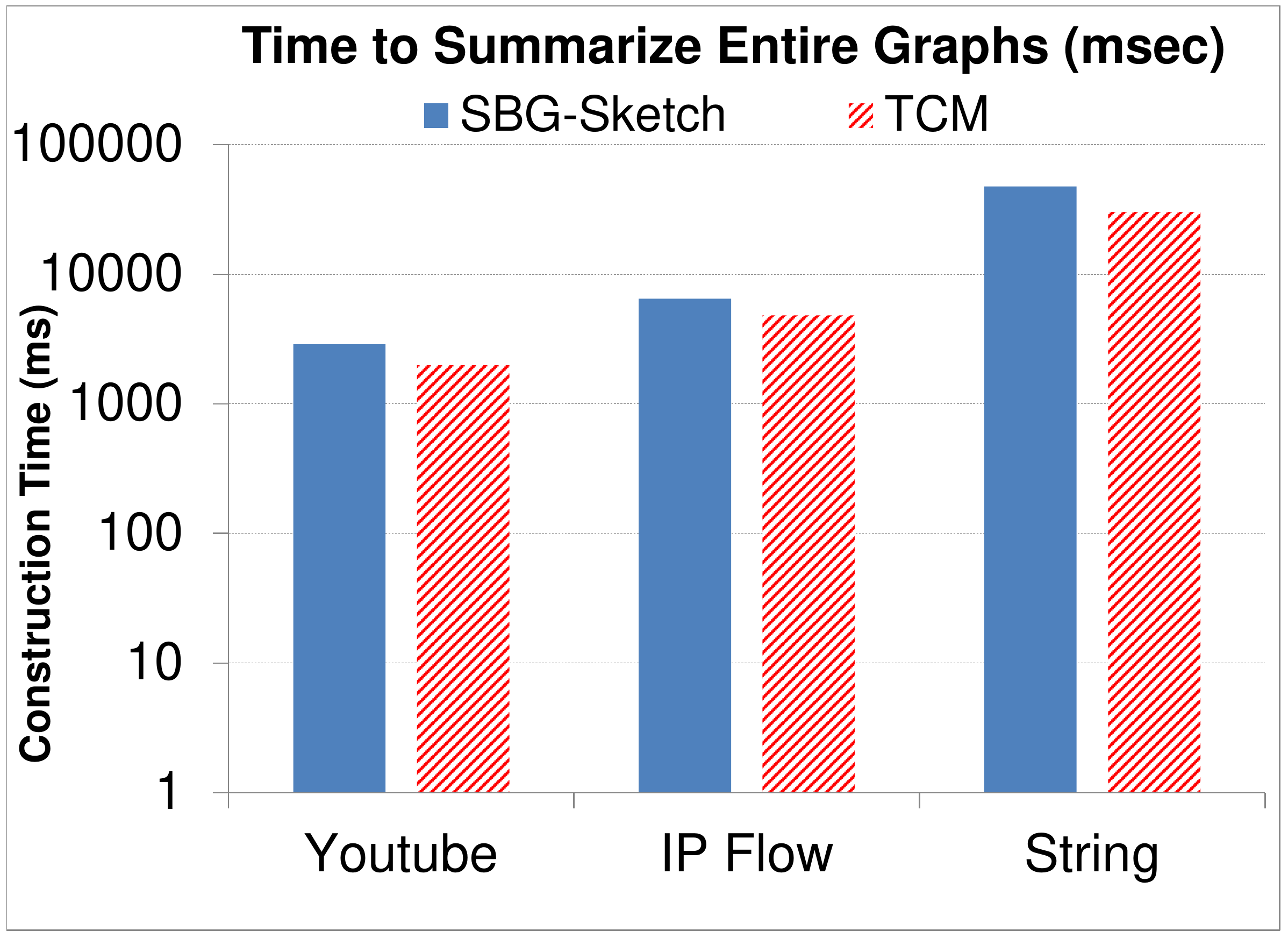}
\caption{\ourSketch{} has an average increase of 28\% in sketch-construction time comparing to TCM.}
\label{Fig:Performance_Construction}
\end{figure}

\section{Related Work}
\label{sec:RelatedWork}

The related work to \ourSketch{} can be divided into two categories: (1)~sketches for general streams, and (2)~sketches for graph streams. In the first category, various research work has been proposed
%until recently
, e.g., Ada-sketch~\cite{Cat1_AdaSketches}, CountMin~\cite{CountMin_Algo_2005}, AMS\cite{Cat1_AMS}, Bottom-k~\cite{Cat1_Bottomk}, and Lossy-Counting\cite{Cat1_LossyCounting}. However, the research efforts of the first category are not optimized for graph streams (see~\cite{gSketch_VLDB_2011}).
\ourSketch{}, our proposed method, is designed to summarize labeled-graph streams effectively.
It is important to note that the eviction ranking mechanism of \ourSketch{} is not related to set membership sketches, e.g., Bloom filters~\cite{kirsch2008less}. Bloom filters do not perform evictions or have rankings.

In the second category, the research efforts focus on processing graph queries over data streams that form graph structures (e.g., ~\cite{Cat2_Triangles,Cat2_SP,gSketch_VLDB_2011,TCM_SIGMOD_2016}).
In~\cite{Cat2_Triangles}, graph queries that count the
number of triangles are addressed, and~\cite{Cat2_SP} supports shortest-path queries. However, both~\cite{Cat2_SP,Cat2_Triangles} and similar theoretical work (e.g.,~\cite{Cat2_DistinctCountNodes}) focus on providing theoretical bounds that may not scale for large graphs.
gSketch~\cite{gSketch_VLDB_2011} extends the idea of the Count-Min sketch~\cite{CountMin_Algo_2005} to compute edge-frequency queries. To construct a sketch, gSketch requires either a sample of the graph stream or both a graph-stream sample and a query-workload sample. gSketch considers only unlabeled-graph streams. In contrast, \ourSketch{} neither requires edge samples nor query-workload samples to summarize labeled-graph streams. In addition, \ourSketch{} supports graph-traversal queries that are not considered by gSketch for its supported graph model. Notice that this category does not consider the graph summarization techniques that are not designed for streaming scenarios (e.g.,~\cite{Navlakha2008,GS_2001_Adler,GS_2001_Suel,GS_2003_Raghavan,GS_Book,GS_2012_Fan}). The reason is that these techniques do not support the continuous arrival of edges in streaming applications as discussed in~\cite{TCM_SIGMOD_2016}.

The most related work to \ourSketch{} is TCM~\cite{TCM_SIGMOD_2016}. The main motivation of TCM is to support graph-traversal queries. TCM builds $K$ independent matrices, where each matrix has two dimensions. Each matrix uses an independent hash function to summarize the graph stream (i.e., the graph summary is created $K$ times with different hash functions). A cell in a TCM sketch is addressed by the endpoints of a given edge to update the sketch on edge arrivals to summarize the graph topology along with an edge attribute. However, TCM is not optimized to handle labeled-graph streams. \cite{TCM_SIGMOD_2016} describes without evaluation how TCM can handle graphs with different type of edges (i.e., labeled-graph streams). In particular, \cite{TCM_SIGMOD_2016} suggests to create a matrix for each edge type. However, this approach does not handle the common edge-skewness w.r.t. the edge labels. Moreover, the edge-skewness may not be known beforehand, and may change with time to make allocating different memory sizes for each label impractical. In contrast, \ourSketch{} handles labeled-graph streams efficiently by reducing the error rate of {\em TCM} by up to $99$\%. Moreover, \ourSketch{} does not require any pre-knowledge about the edge distribution. 

\section{Conclusion}
\label{sec:Conclusion}

\ourSketch{} is a graphical sketching method that summarizes labeled graph streams, where the graph topology is considered in the summary. It assumes a stream, where each edge has one label. \ourSketch{} addresses the consequences of having unbalanced edge-distribution w.r.t. the edge labels. This is achieved by presenting and evaluating a ranking technique. Given a fixed sketch-size, the proposed ranking technique allows \ourSketch{} to automatically adapt to the unbalanced labels of the streamed edges by allowing an edge to use more than one matrix based on its ranks. Moreover, it guarantees that all the edges gain in summarization accuracy even if their labels are relatively-rare. We demonstrate how \ourSketch{} can be used to approximate several graph-query types that depend on an aggregation of an edge attribute and/or the topology of the graph. The experimental study over three real labeled-graphs spanning different domains show that \ourSketch{} reduces the estimation error of the state-of-the-art by up to $99$\%.

%\cleardoublepage
%\pagebreak
\begin{small}
\bibliographystyle{abbrv}
\bibliography{references}
\end{small}

\appendix

%Appendix A
\section{Proofs}

%Appendix A.1
\subsection{Proof of Theorem~\ref{t:boundEdge}}
\label{sec:Proof_BoundEdge}
\begin{nonumtheorem}
Let $L$ be the number of priorities and $X_e$ be the number of arrivals of edge $e \in V\times V$ during an observation window, where $V$ is the set of vertexes in the graph stream. 
Let $P \geq 1$ be the number of $P$-independent hash functions used in \ourSketch{}.
Let $1/(1+\alpha)$ be the percentage reduction in the number of sketch counter cells due to the inclusion of \ourSketch priority counters.
Furthermore, 
let $X_e^\text{(\ourSketch)} \geq X_e$ be the upper bound on $X_e$ given by \ourSketch{} and let $X_e^\text{(TCM)} \geq X_e$ be absolute error distribution given by TCM with the same number of sketch counter cells. Assume that 
edges arrive according to a nonhomogeneous Poisson process (a Poisson process that varies over time) with average rate $\lambda_e$ over the observation window.
The observation window is defined to be of length one by an appropriate change of units. 
Then, the distribution of the absolute error is
\begin{align*}
&\text{Pr}[X_e^\text{(\ourSketch)} - X_e > k]  < \left(\text{Pr}[X_e^\text{(TCM)} - X_e > k ]   - \zeta_{k,L,P} \right)^P \, ,
\end{align*}
where 
\begin{align*}
& \zeta_{k,L,P} = \sum_{j=k+1}^{K_\text{max}}\frac{((1 + \alpha)P\tilde{\lambda}_0)^j}{j!} \exp\left( -(1 + \alpha)P\tilde{\lambda}_0\right) \\
 &\times \Bigg( 1 - \prod_{i=1}^{L-1} \bigg( 1 - \sum_{k_i=0}^{k} \left(\frac{1}{L-1}\right)^{k_i} \left(\frac{L-1-i}{L-1}\right)^{j - k_i } \binom{j}{k_i} \\ 
&\qquad\qquad\qquad\times  \exp\left( - (i+1) \frac{P\tilde{\lambda} (1+\alpha)}{L} \right) \bigg) \Bigg)
\end{align*}
and
$\tilde{\lambda}_0 = d^{-2}\sum_{h \in V\times V} \lambda_h {\bf 1}\{\text{$h$ has same label as $e$}\}$, $\tilde{\lambda} = d^{-2}\left(\sum_{h \in V \times V } \lambda_h - \lambda_0 \right)$, $K_\text{max} \gg k$ is an arbitrary constant, and {\bf 1} is the Kronecker delta function.
\end{nonumtheorem}
\begin{proof}
In what follows we say an edge $e'$ has ``higher priority'' than an edge $e^\star$ at sketch matrix $M$ if the priority number of $e'$ is smaller than that of $e^\star$ in $M$.
We start with the case $P=1$, one hash function.
The number of arrivals of an edge $e \in V\times V$ in the observed time window is a Poisson distributed random variable $X_e \sim \text{Poisson}(\lambda_e)$.
Let $X^{(\ourSketch)}_e$ be upper bound on $X_e$ returned by \ourSketch{}. In what follows we condition on edge $e$ having at least one arrival $X_e > 0$ in the observation time window. 
Note that the difference $X_e^\text{(\ourSketch)} - X_e$ is due to the collision between $e$ and other edges.
Without loss of generality we define $M_0$ to be the matrix that edge $e$ has priority 0, i.e., the matrix of label of edge $e$. 
Let $X^{(TCM)}_e$ be upper bound on $X_e$ returned by TCM assuming a $(1+\alpha)$ increase in counter load:
\[
\text{Pr}[X_e^\text{(TCM)} - X_e \leq k ] = \sum_{j=0}^{k}\frac{((1 + \alpha)p\tilde{\lambda}_0)^j}{j!} \exp\left( -(1 + \alpha)W\tilde{\lambda}_0\right)
\]
Note that the probability $\text{Pr}[X_e^\text{(\ourSketch)} - X_e \leq k] $ 
is the probability that the arrivals in $M_0$ are at most $k$, $\text{Pr}[X_e^\text{(TCM)} - X_e \leq k ]$, or there are more than $k$ arrivals and these extra edge arrivals are distributed into the matrices of other labels $M_1,\ldots,M_L$.
Without loss of generality let $M_i$ be the sketch where edge $e$ has priority $i$. The probability that the counter values will have values less than $k$ in $M_i$, $i \in \{1,\ldots,L\}$ from $j > k$ arrivals at $M_0$ is
\begin{equation}\label{e:2}
\gamma_i = \sum_{k_i=0}^{k} \left(\frac{1}{L-1}\right)^{k_i}\left(\frac{L-1-i}{L-1}\right)^{j - k_i } 
 \frac{j!}{k_i! (j - k_i)!} \, .
\end{equation}
The probability that some $M_i$ will have less than $k$ collision is then $1- \prod_{i=1}^{L-1}(1-\gamma_i)$.
The probability the counter containing $e$ survives an eviction from higher priority edges is
\begin{equation}\label{e:1}
\text{Pr}[e \text{ is not evicted from sketch }M_i] = \exp\left( - i \frac{\tilde{\lambda} (1+\alpha)}{L}\right) ,
\end{equation}
where $\tilde{\lambda} = \left(\sum_{h \in V \times V } \lambda_h - \lambda_0 \right)/d^2$ is the rate of all edge arrivals except edges with the same label as edge $e$.   
A same priority edge can also collide with $e$ at $M_i$. While this does not mean there will be more than $k$ collisions with $e$, we just assume we do not want any further collisions to get a lower bound, multiplying the above by $\exp\left( - \frac{\tilde{\lambda} (1+\alpha)}{L}\right)$.
Collecting all the terms we get the equation for $P=1$ hash functions.

To consider $p \geq 1$ hash functions, we observe that having $p$ hash functions also increases the arrival rate per counter, multiplying it by $p$.
On the other hand, because we assume the hash functions are $p$-independent, because $X_e^\text{(\ourSketch)}$ is the minimum value over all the sketches of $p$-independent different hash functions, the probability that for all the hash functions we have $X_e^\text{(\ourSketch)} - X_e > k$ is  $\text{Pr}[X_e^\text{(\ourSketch)} - X_e > k]^P$, which concludes our proof.
\end{proof}

%Appendix A.2
\subsection{Proof of Theorem~\ref{Theory:zeroEstimation}}
\label{sec:Proof_ZeroEstimate}
%We repeat the theorem statement for convenience.
%\begin{nonumtheorem}
%{\em Using \ourSketch{}, ~$\hat{f}_{e}(a,~b,~L_i)~=~0$ $\implies$ $f_{e}(a,~b,~L_i)~=~0$.}
%\end{nonumtheorem}
\begin{proof}
The proof is by contradiction. Assume that Edge~$E = (a,~b,~L_i)$ was inserted into \ourSketch{}. Then, all the candidate cells of Edge~$E$ have ranks that are either equal to or higher than the corresponding ranks of Edge~$E$ (see Lines~$7-12$ of Algorithm~\ref{Algo:UpdateEdgeQuerySketch}). So, when the edge-query estimator hits a cell with a rank that is lower than the corresponding rank of Edge~$E$, then this contradicts that Edge~$E$ was received before.
\end{proof}

\end{document}